%% file: main.tex
\documentclass[sigconf,nonacm]{acmart}
\setcopyright{none}
\renewcommand\footnotetextcopyrightpermission[1]{}
\usepackage{adjustbox,booktabs,algorithm,algpseudocode,xspace,multirow,subcaption,listings,tikz}
\usetikzlibrary{positioning,arrows.meta,decorations.pathreplacing,calc,fit,patterns,backgrounds}
\input{macros}
\input{claims_macros}
\begin{document}
\title{Semantic Views: Certified Reuse of LLM Labels Across Semantic Queries}
\author{Oktie Hassanzadeh}
\affiliation{\institution{IBM}\country{}}
\email{hassanzadeh@us.ibm.com}
\author{Dharmashankar Subramanian}
\affiliation{\institution{IBM}\country{}}
\email{dharmash@us.ibm.com}

\begin{abstract}
\input{sections/abstract}
\end{abstract}


\maketitle

\input{sections/introduction}
\input{sections/background}
\input{sections/views}
\input{sections/reuse}
\input{sections/maintenance}
\input{sections/system}
\input{sections/evaluation}
\input{sections/discussion}
\input{sections/related}
\input{sections/conclusion}

\bibliographystyle{ACM-Reference-Format}
\bibliography{references}

\appendix
\input{appendix}
\end{document}

%% file: macros.tex
\newcommand{\sys}{\textsc{SemViews}\xspace}
\newcommand{\lotus}{LOTUS\xspace}
\newcommand{\cf}{\textsc{CertifiedFilter}\xspace}
\newcommand{\eps}{\varepsilon}
\newcommand{\sel}{\mathrm{sel}}
\newcommand{\contained}[1]{\sqsubseteq_{#1}}

\theoremstyle{acmdefinition}
\newtheorem{problem}{Problem}

%% file: claims_macros.tex
\newcommand{\claimEOneAurocDbpedia}{0.95\xspace}
\newcommand{\claimEOneAurocGoemotions}{0.95\xspace}
\newcommand{\claimEOneAurocPubmed}{0.87\xspace}
\newcommand{\claimEOneAurocReviews}{0.82\xspace}
\newcommand{\claimEOneContainZeroOnePairs}{92\xspace}
\newcommand{\claimEOneJudgedImplies}{146\xspace}
\newcommand{\claimEOneJudgedImpliesBadShare}{56\%\xspace}
\newcommand{\claimEOnePairs}{5,280\xspace}
\newcommand{\claimEOneZeroAdversarialCostRatioMax}{0.98\xspace}
\newcommand{\claimEOneZeroAdversarialMetMin}{99.0\%\xspace}
\newcommand{\claimEOneZeroInvertedSavingMax}{40\%\xspace}
\newcommand{\claimEOneZeroInvertedSavingMin}{16\%\xspace}
\newcommand{\claimEOneZeroSimcacheMet}{0\%\xspace}
\newcommand{\claimEOneOneAppendMet}{100.0\%\xspace}
\newcommand{\claimEOneOneAppendShareMax}{49\%\xspace}
\newcommand{\claimEOneOneAppendShareMin}{32\%\xspace}
\newcommand{\claimEOneOneSimcacheMet}{0\%\xspace}
\newcommand{\claimEOneTwoStrictLotusMissMax}{48.0\%\xspace}
\newcommand{\claimEOneTwoStrictLotusMissMin}{1.3\%\xspace}
\newcommand{\claimEOneTwoStrictLotusValidDatasets}{1\xspace}
\newcommand{\claimEOneThreeHumanLlmAgreeMedian}{97.7\%\xspace}
\newcommand{\claimEOneThreeHumanLlmFOneMedian}{0.65\xspace}
\newcommand{\claimEOneThreeHumanMaxMiss}{0.6\%\xspace}

\newcommand{\claimEOneThreeHumanOrders}{4\xspace}
\newcommand{\claimEOneThreeHumanPreds}{40\xspace}

\newcommand{\claimEOneThreeHumanShareLarge}{20\%\xspace}
\newcommand{\claimEOneThreeHumanShareSmall}{52\%\xspace}
\newcommand{\claimEOneThreeLlmDbpediaMaxMiss}{1.2\%\xspace}

\newcommand{\claimEOneThreeLlmDbpediaOrders}{20\xspace}
\newcommand{\claimEOneThreeLlmDbpediaPreds}{8\xspace}

\newcommand{\claimEOneThreeLlmDbpediaSelMax}{19.6\%\xspace}
\newcommand{\claimEOneThreeLlmDbpediaSelMin}{0.9\%\xspace}
\newcommand{\claimEOneThreeLlmDbpediaShareLarge}{34\%\xspace}
\newcommand{\claimEOneThreeLlmDbpediaShareSmall}{60\%\xspace}
\newcommand{\claimEOneThreeLlmLabels}{1,600,000\xspace}
\newcommand{\claimEOneThreeLlmReviewsMaxMiss}{0.6\%\xspace}

\newcommand{\claimEOneThreeLlmReviewsRareShareLarge}{66\%\xspace}
\newcommand{\claimEOneThreeLlmReviewsRareShareSmall}{95\%\xspace}
\newcommand{\claimEOneThreeLlmReviewsSelMax}{75.0\%\xspace}
\newcommand{\claimEOneThreeLlmReviewsSelMin}{0.4\%\xspace}
\newcommand{\claimEOneThreeLlmReviewsShareLarge}{55\%\xspace}
\newcommand{\claimEOneThreeLlmReviewsShareSmall}{77\%\xspace}
\newcommand{\claimEOneThreeRows}{100,000\xspace}

\newcommand{\claimETwoBargainSavingMax}{10\%\xspace}
\newcommand{\claimETwoBargainSavingMin}{2\%\xspace}
\newcommand{\claimETwoBreakdownCertificationMax}{32\%\xspace}
\newcommand{\claimETwoBreakdownCertificationMin}{18\%\xspace}
\newcommand{\claimETwoBreakdownGapClosingMax}{20\%\xspace}
\newcommand{\claimETwoBreakdownGapClosingMin}{6\%\xspace}
\newcommand{\claimETwoBreakdownOracleRegionMax}{11\%\xspace}
\newcommand{\claimETwoBreakdownOracleRegionMin}{6\%\xspace}
\newcommand{\claimETwoBreakdownRefitRoundsMax}{54\%\xspace}
\newcommand{\claimETwoBreakdownRefitRoundsMin}{41\%\xspace}
\newcommand{\claimETwoBreakdownShapingMax}{7\%\xspace}
\newcommand{\claimETwoBreakdownShapingMin}{6\%\xspace}
\newcommand{\claimETwoCloseStepsMedianMax}{4\xspace}
\newcommand{\claimETwoCloseStepsPNineZeroMax}{41\xspace}
\newcommand{\claimETwoCloseStepsZeroShare}{54\%\xspace}

\newcommand{\claimETwoLateShareMin}{52\%\xspace}
\newcommand{\claimETwoLateViewGainMax}{31\%\xspace}
\newcommand{\claimETwoLateVsBargainMax}{1.89$\times$\xspace}
\newcommand{\claimETwoLateVsBargainMin}{1.32$\times$\xspace}

\newcommand{\claimETwoRatioValidMax}{1.32$\times$\xspace}
\newcommand{\claimETwoRatioValidMin}{1.04$\times$\xspace}
\newcommand{\claimETwoRefineCutInvMax}{1.27$\times$\xspace}
\newcommand{\claimETwoRefineCutInvMin}{1.09$\times$\xspace}
\newcommand{\claimETwoRefineCutMax}{21\%\xspace}
\newcommand{\claimETwoRefineCutMin}{8\%\xspace}
\newcommand{\claimETwoRefitRoundsMedian}{6\xspace}
\newcommand{\claimETwoRefitRoundsPNineZero}{8\xspace}

\newcommand{\claimETwoSavingRatioMax}{18$\times$\xspace}
\newcommand{\claimETwoSavingRatioMin}{4$\times$\xspace}
\newcommand{\claimETwoSemviewsSavingMax}{45\%\xspace}
\newcommand{\claimETwoSemviewsSavingMin}{28\%\xspace}

\newcommand{\claimETwoSemweaveShareMax}{91\%\xspace}
\newcommand{\claimETwoSemweaveShareMin}{49\%\xspace}

\newcommand{\claimETwoUsdCallGap}{0.7\xspace}
\newcommand{\claimETwoViewGainMax}{25\%\xspace}
\newcommand{\claimETwoViewGainMin}{4\%\xspace}
\newcommand{\claimETwoVsBargainMax}{1.71$\times$\xspace}
\newcommand{\claimETwoVsBargainMin}{1.32$\times$\xspace}

\newcommand{\claimETwoVsBargainViewsVOneMax}{1.18$\times$\xspace}
\newcommand{\claimETwoVsBargainViewsVOneMin}{0.99$\times$\xspace}
\newcommand{\claimEThreeBargainMaxMiss}{0.6\%\xspace}
\newcommand{\claimEThreeCheaperThanSemviewsMinMiss}{32.4\%\xspace}
\newcommand{\claimEThreeDelta}{0.05\xspace}

\newcommand{\claimEThreeLotusMaxMiss}{65.4\%\xspace}
\newcommand{\claimEThreeLotusMinMiss}{30.8\%\xspace}
\newcommand{\claimEThreeLotusPredicates}{145\xspace}
\newcommand{\claimEThreeLotusPredicatesAboveDelta}{102\xspace}
\newcommand{\claimEThreeLotusCascadeFullMaxMiss}{57.5\%\xspace}
\newcommand{\claimEThreeProxyLrMaxMiss}{65.4\%\xspace}
\newcommand{\claimEThreeSemviewsMaxMiss}{0.7\%\xspace}
\newcommand{\claimEThreeSemviewsMsPerQuery}{849\xspace}
\newcommand{\claimEThreeSemviewsPooledMet}{99.63\%\xspace}
\newcommand{\claimEThreeSemviewsPooledMiss}{0.37\%\xspace}
\newcommand{\claimEThreeSemviewsPredicates}{145\xspace}
\newcommand{\claimEThreeSemviewsPredicatesAboveDelta}{0\xspace}
\newcommand{\claimEThreeSemviewsQueries}{145,000\xspace}
\newcommand{\claimEThreeSemviewsSeeds}{1,000\xspace}
\newcommand{\claimEThreeSemviewsWorstPredicateMiss}{3.2\%\xspace}
\newcommand{\claimEThreeSemviewsWorstPredicateUpper}{4.5\%\xspace}
\newcommand{\claimEThreeSemweaveMaxMiss}{67.8\%\xspace}
\newcommand{\claimEThreeSemweaveMinMiss}{6.2\%\xspace}
\newcommand{\claimEThreeSemweaveMissGivenReuse}{65\%\xspace}
\newcommand{\claimEThreeSemweaveRecallPFive}{0.22\xspace}
\newcommand{\claimEThreeSemweaveReuseShare}{56\%\xspace}

\newcommand{\claimEThreeSimcacheMaxMiss}{80\%\xspace}
\newcommand{\claimEThreeSimcacheMinMiss}{31\%\xspace}
\newcommand{\claimEThreeSimcacheMissGivenHit}{100\%\xspace}
\newcommand{\claimEFourRowsLargeShare}{57\%\xspace}
\newcommand{\claimEFourRowsSmallShare}{81\%\xspace}
\newcommand{\claimEFiveGoemotionsHalfSaving}{13\%\xspace}
\newcommand{\claimEFivePartialOtherMaxDev}{3.6\%\xspace}
\newcommand{\claimESevenCrossModelMet}{100\%\xspace}
\newcommand{\claimESevenCrossModelSaving}{18\%\xspace}
\newcommand{\claimESevenSimcacheCrossMet}{0\%\xspace}
\newcommand{\claimENineWorstOverhead}{0.0\%\xspace}
\newcommand{\claimFloorCoef}{33\xspace}
\newcommand{\claimFloorDrawsAtOnePct}{3,300\xspace}
\newcommand{\claimHeadroomOperatorShare}{60\%\xspace}
\newcommand{\claimHeadroomPerfectCommonShare}{14\%\xspace}
\newcommand{\claimHeadroomPerfectRareShare}{80\%\xspace}
\newcommand{\claimHeadroomPerfectShare}{31\%\xspace}
\newcommand{\claimHeadroomScaleLargePerfectShare}{7\%\xspace}
\newcommand{\claimHeadroomScaleLargeShapedShare}{48\%\xspace}
\newcommand{\claimHeadroomScaleLargeShare}{51\%\xspace}

\newcommand{\claimHeadroomScaleMidShare}{52\%\xspace}
\newcommand{\claimHeadroomWarmExtraMax}{16\%\xspace}
\newcommand{\claimHeadroomWarmExtraMin}{5\%\xspace}
\newcommand{\claimNoiseAgreement}{99.85\%\xspace}
\newcommand{\claimSelHighShare}{54\%\xspace}
\newcommand{\claimSelLowPreds}{44\xspace}
\newcommand{\claimSelLowShare}{80\%\xspace}
\newcommand{\claimSelPreds}{145\xspace}

\newcommand{\claimTapeOracleCalls}{725,000\xspace}
\newcommand{\claimViewsPartnersGoemotions}{25\xspace}
\newcommand{\claimViewsPartnersOtherMax}{9\xspace}
\newcommand{\claimViewsRelatedGainMax}{38\%\xspace}
\newcommand{\claimViewsRelatedGainMin}{11\%\xspace}

\newcommand{\claimViewsUnrelatedGainMax}{14\%\xspace}
\newcommand{\claimViewsUnrelatedGainOtherMax}{5\%\xspace}
\newcommand{\claimViewsVOneRelatedGainMax}{31\%\xspace}
\newcommand{\claimViewsVOneRelatedGainMin}{2\%\xspace}
\newcommand{\claimWrAgreeMax}{99.8\%\xspace}
\newcommand{\claimWrAgreeMin}{76\%\xspace}
\newcommand{\claimWrJaccardMax}{0.95\xspace}
\newcommand{\claimWrJaccardMin}{0.57\xspace}
\newcommand{\claimWrOrdersCost}{20\xspace}
\newcommand{\claimWrOrdersValid}{200\xspace}

\newcommand{\claimWrSemviewsMetMin}{99.8\%\xspace}
\newcommand{\claimWrSemviewsMissMax}{0.2\%\xspace}
\newcommand{\claimWrSemviewsShareMax}{53\%\xspace}
\newcommand{\claimWrSemviewsShareMin}{47\%\xspace}

\newcommand{\claimWrTwins}{15\xspace}
\newcommand{\claimWrUnsafeMetMax}{47\%\xspace}
\newcommand{\claimWrUnsafeMetMin}{17\%\xspace}
\newcommand{\claimWrUnsafeMissMax}{83\%\xspace}
\newcommand{\claimWrUnsafeMissMin}{53\%\xspace}
\newcommand{\claimWrViewGainMax}{36\%\xspace}
\newcommand{\claimWrViewGainMin}{27\%\xspace}

%% file: sections/abstract.tex
Semantic query engines evaluate natural-language filters with a large language model, and every
evaluation is a costly oracle call. Optimizers that reduce these calls face a dilemma: cascades
that trust a cheap proxy miss their precision and recall targets on many queries, methods with
finite-sample guarantees save little, and all of them discard the oracle's labels when the query
ends. We present \sys, a certified semantic filter that orders tuples by signals that cost no
oracle calls (the labels that earlier queries left behind as \emph{semantic views}, a local
embedding, and a helper model's score), refines that order with labels it must pay for anyway,
and certifies its answer with fresh stratified oracle samples. The answer meets user precision and
recall targets with probability at least $1-\delta$ for every workload and every catalog, so a stale
or misleading view can cost oracle calls but never accuracy. On four workloads replayed against
\claimTapeOracleCalls{} recorded oracle labels, \sys meets both targets on
\claimEThreeSemviewsPooledMet{} of queries and saves
\claimETwoSemviewsSavingMin{}--\claimETwoSemviewsSavingMax{} of the oracle calls, whereas the only
prior method in our study with the same guarantee needs
\claimETwoVsBargainMin{}--\claimETwoVsBargainMax{} as many calls. When a question returns in
different words, views cut its calls by a further \claimWrViewGainMin{}--\claimWrViewGainMax{},
with both targets met on at least \claimWrSemviewsMetMin{} of such queries; similarity caches and
entailment-based reuse meet them on only
\claimWrUnsafeMetMin{}--\claimWrUnsafeMetMax{}.

%% file: sections/introduction.tex
\section{Introduction}
\label{sec:intro}

Semantic query engines let users filter, join, and rank data with natural-language predicates
that a large language model (LLM) evaluates tuple by
tuple~\cite{patel2025lotus,liu2024palimpzest,dai2024uqe}, and cloud data warehouses now expose
the same capability as AI functions in SQL~\cite{liskowski2026cortex}. Each evaluation is an
\emph{oracle call}, and oracle calls dominate the cost and latency of these queries. A growing body
of work therefore optimizes one semantic operator at a time. Model cascades route easy tuples to
a cheap helper model and still meet precision and recall
guarantees~\cite{patel2025lotus,kang2020supg,zeighami2025bargain}; proxies, clusters, and learned
orderings cut the oracle's share further~\cite{chung2026proxy,hou2026csv,zhao2026larch}. These
methods face a dilemma. Cascades that trust a proxy's calibration miss their targets on a large
share of queries: on our workloads, the \lotus cascade misses on
\claimEThreeLotusMinMiss{}--\claimEThreeLotusMaxMiss{} of them. Methods with finite-sample
guarantees keep their targets, but with a weak proxy they save little:
BARGAIN~\cite{zeighami2025bargain} with the helper model of the \lotus cascade saves only
\claimETwoBargainSavingMin{}--\claimETwoBargainSavingMax{} of the oracle calls. Across queries,
these methods reuse at most query-independent artifacts, such as embedding
indexes~\cite{kang2022tasti}, or thresholds and labels for the \emph{same} predicate. The labels
paid for one predicate are never used to answer another.

\begin{example}[Running example]
\label{ex:running}
An analyst explores a table of product reviews and first asks for reviews that ``complain about
the battery'', then for reviews that ``report a hardware defect'', and then for reviews that ``say
the battery overheats''. The third predicate is narrower than the first, and the second overlaps
both. An exact-match cache never hits, because no two predicates are the same. Yet the labels
already paid for say a great deal about the third question: a review the oracle judged not to
complain about the battery is unlikely to say that the battery overheats.
\end{example}

\noindent
Semantic queries rarely arrive alone. Analysts refine predicates during exploration, and agents
issue such refinements at machine speed~\cite{liu2026agentfirst}. This paper asks whether the
labels of earlier predicates can help answer new, different predicates \emph{without weakening
the guarantee} that the best per-query optimizers give.

\paragraph{Reuse is tempting, but unsafe.}
Two kinds of reuse suggest themselves. A similarity cache answers a new predicate with the stored
answer of a similar one~\cite{gptcache}. Systems that infer containment between predicates, by
logic or by an LLM's entailment judgment~\cite{mahmood2026semweave,xu2022eva}, reuse an answer
when one predicate seems to imply another. Both assume that relations between predicates are
crisp. On the oracle's actual labels, they are not. Across our four workloads,
\claimEOneJudgedImpliesBadShare{} of the implications that an LLM judge asserts are violated by
more than 10\% of the antecedent's positives, and containment errors compound along chains of
predicates (Section~\ref{sec:views}). As a result, the answers that similarity caches serve on a hit miss a precision or
recall target in \claimEThreeSimcacheMissGivenHit{} of cases in our workloads, and SemWeave's entailment-based reuse~\cite{mahmood2026semweave} misses on up to
\claimEThreeSemweaveMaxMiss{} of queries (Section~\ref{sec:eval-validity}).

\paragraph{Our approach.}
\sys is a certified filter that uses everything it has for free to decide where to look, and only
fresh oracle samples to decide what is true (Figure~\ref{fig:running}). Three signals cost no oracle
calls: a helper model's score, a local embedding of each tuple, and, new in \sys, the labels that
earlier queries left behind. We keep those labels as \emph{semantic views}: label columns bound to
a predicate, the columns it reads, and an oracle model, possibly partial, with the provenance of
every label. \sys fits a score for the new predicate on these signals, refines it with the labels
of the tuples it is least sure about, which it would have to pay for anyway, and splits the score
order into a reject region, an oracle region, and an accept region. \emph{Fresh} oracle samples,
drawn uniformly after the regions are fixed, certify both outer regions with exact finite-sample
bounds; the oracle region is labeled in full. The guarantee therefore holds for every workload and
every catalog: a stale, partial, or misleading view costs oracle calls, never correctness, and
reuse across queries is safe by construction.

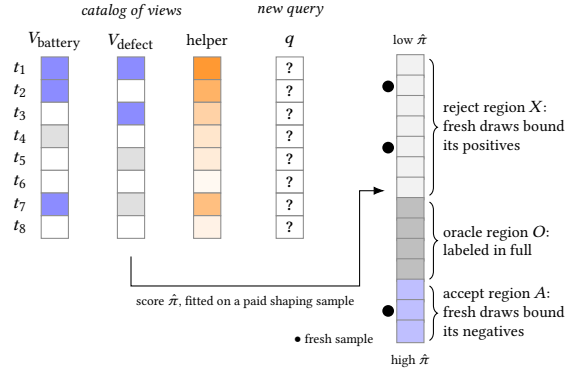
\begin{figure}[t]
\centering
\begin{tikzpicture}[x=0.42cm,y=0.30cm,font=\scriptsize,
  cell/.style={draw=black!45,minimum width=0.36cm,minimum height=0.30cm,inner sep=0pt},
  pcell/.style={cell,fill=blue!45}, ncell/.style={cell,fill=white}, ucell/.style={cell,fill=black!12}]
  \foreach \r/\a/\b in {1/pcell/pcell,2/pcell/ncell,3/ncell/pcell,4/ucell/ncell,5/ncell/ucell,6/ncell/ncell,7/pcell/ucell,8/ncell/ncell}{
    \node[\a] at (0,-\r) {}; \node[\b] at (2.4,-\r) {};
    \node[anchor=east] at (-0.6,-\r) {$t_{\r}$};
  }
  \foreach \r/\s in {1/80,2/60,3/35,4/20,5/15,6/5,7/50,8/10}{ \node[cell,fill=orange!\s] at (4.8,-\r) {}; }
  \foreach \r in {1,...,8}{ \node[cell] at (7.4,-\r) {\textbf{?}}; }
  \node at (0,0.2) {$V_{\text{battery}}$};
  \node at (2.4,0.2) {$V_{\text{defect}}$};
  \node at (4.8,0.2) {helper};
  \node at (7.4,0.2) {$q$};
  \node[font=\scriptsize\itshape] at (2.4,1.5) {catalog of views};
  \node[font=\scriptsize\itshape] at (7.4,1.5) {new query};
  \draw[-{Latex[length=1.5mm]}] (2.4,-9.3) -- (2.4,-10.6) -| (9.6,-6.4)
     node[pos=0.25,below,align=center,font=\tiny]{score $\hat\pi$, fitted on a paid shaping sample} -- (10.4,-6.4);
  \begin{scope}[shift={(11.2,0)}]
    \foreach \i in {1,...,7}{ \node[cell,fill=black!5] at (0,-\i*0.9) {}; }
    \foreach \i in {8,...,11}{ \node[cell,fill=black!25] at (0,-\i*0.9) {}; }
    \foreach \i in {12,...,14}{ \node[cell,fill=blue!25] at (0,-\i*0.9) {}; }
  \end{scope}
  \foreach \y in {-1.8,-4.5,-11.7}{ \fill (10.5,\y) circle (2.2pt); }
  \draw[decorate,decoration={brace,amplitude=3pt}] (11.8,-0.6) -- (11.8,-6.6)
     node[midway,right=2pt,align=left]{reject region $X$:\\fresh draws bound\\its positives};
  \draw[decorate,decoration={brace,amplitude=3pt}] (11.8,-6.9) -- (11.8,-10.3)
     node[midway,right=2pt,align=left]{oracle region $O$:\\labeled in full};
  \draw[decorate,decoration={brace,amplitude=3pt}] (11.8,-10.6) -- (11.8,-13.0)
     node[midway,right=2pt,align=left]{accept region $A$:\\fresh draws bound\\its negatives};
  \node[font=\tiny] at (11.2,0.3) {low $\hat\pi$};
  \node[font=\tiny] at (11.2,-13.9) {high $\hat\pi$};
  \node[font=\tiny,anchor=east] at (10.3,-13.0) {$\bullet$ fresh sample};
\end{tikzpicture}
\caption{Certified reuse on the running example. Past views (blue: positive, white: negative,
gray: unknown) and the helper score are features of a score $\hat\pi$ for the new predicate $q$.
The planner splits the score order into three regions. Only fresh oracle samples drawn after the
split enter the confidence bounds, so wrong views change the cost, never the guarantee.}
\label{fig:running}
\end{figure}

\paragraph{Technical challenges.}
Certified reuse must solve four problems that per-query optimizers do not face.
\emph{(1) Untrusted side information.} A view of another predicate can be arbitrarily wrong about
the new one, so no bound may depend on view labels being right.
\emph{(2) Biased labels.} The labels in a view were chosen by an earlier query's plan, typically
the tuples that were uncertain for a different predicate. They are not a sample of anything for
the new predicate, and an estimator that treats them as one is biased.
\emph{(3) Data-dependent design.} To profit from views, the regions must depend on the data, yet
finite-sample bounds are valid only when the regions are fixed before the draws that certify
them.
\emph{(4) Planning under uncertainty.} Choosing regions and sample sizes requires predicting
whether certification will succeed. A wrong prediction should cost a few extra calls, not
validity, and should not force the query to fall back to labeling every tuple.

\cf, the operator at the core of \sys, addresses all four by separating \emph{ordering} from
\emph{certification}. Views, helper scores, and embeddings enter only a score fitted on paid labels: a small uniform shaping sample,
refined with labels of the tuples the score is least sure about (challenges 1 and 2). The planner
fixes the regions from that score before any certification draw, and draws without replacement
make exact hypergeometric bounds valid for the finite regions (challenge 3). When the planner's
prediction falls short, a deterministic gap-closing step labels the most uncertain tuples, and a
single coverage event covers every state in which the operator can stop (challenge 4).

\paragraph{Contributions.}
\begin{itemize}
\item We present \cf, a semantic filter that meets user precision and recall targets with
probability $1-\delta$ while ordering tuples by free side information, refining that order with
labels it must pay for anyway, and certifying two regions with exact finite-sample bounds chosen by
a planner that minimizes predicted oracle calls (Section~\ref{sec:reuse}).
\item We make the oracle's labels reusable across queries: semantic views with provenance, used only
to plan, so that reuse cannot break the guarantee for any catalog, including adversarial ones
(Sections~\ref{sec:background}, \ref{sec:views}, and~\ref{sec:maintenance}).
\item We show that uncertified reuse is unsafe: relations between real predicates are statistical
and degrade along chains, more than half of the implications an LLM judges fail on the oracle's
labels, and even the same question asked in different words receives different labels
(Sections~\ref{sec:views} and~\ref{sec:eval}).
\item We implement \sys as a \lotus optimizer that rewrites semantic filters without changing the
engine, and evaluate it on four predicate pools and a workload of reworded questions against
per-query cascades, BARGAIN, exact and similarity caches, and SemWeave's entailment-based reuse,
replayed on recorded oracle labels (\claimTapeOracleCalls{} for the four pools alone;
Sections~\ref{sec:system} and~\ref{sec:eval}).
\end{itemize}

\paragraph{Summary of results.}
\sys missed a target on \claimEThreeSemviewsPooledMiss{} of \claimEThreeSemviewsQueries{} replayed
queries, below $\delta$, and no single predicate exceeded $\delta$. It saved
\claimETwoSemviewsSavingMin{}--\claimETwoSemviewsSavingMax{} of the oracle calls of exact evaluation,
\claimETwoSavingRatioMin{}--\claimETwoSavingRatioMax{} as much as BARGAIN with the helper score, and
every method that is cheaper misses its targets on at least
\claimEThreeCheaperThanSemviewsMinMiss{} of queries. Views
add to these savings in proportion to how related a workload's predicates are: they removed
\claimViewsRelatedGainMin{}--\claimViewsRelatedGainMax{} of the calls of a query when an earlier query
asked a strongly correlated predicate, and \claimWrViewGainMin{}--\claimWrViewGainMax{} when it asked the
same question in different words, where similarity caches and SemWeave miss their targets on
\claimWrUnsafeMissMin{}--\claimWrUnsafeMissMax{} of queries; on unrelated predicates they add no
measurable cost.
Much of the remaining cost is a certification floor that depends on a predicate's selectivity but
not on the size of the relation (Proposition~\ref{prop:floor}), so the share of exact evaluation that
\sys needs falls as tables grow.

\paragraph{Scope and outline.}
We study semantic filters over text; joins, ranking, and multimodal data are out of scope. Our work is orthogonal to methods that make each call cheaper, such as prefix
and KV-cache reuse~\cite{son2026kalypso,liu2024llmqueries}, and to optimizers that allocate error
budgets across the operators of one plan~\cite{sanmartino2026stretto}: \sys shrinks the set of
tuples on which a per-query method has to work. Sections~\ref{sec:background}--\ref{sec:system}
cover the problem, semantic views, \cf, maintenance, and the system;
Sections~\ref{sec:eval}--\ref{sec:related} cover the evaluation, discussion, and related work.

%% file: sections/background.tex
\section{Background and Problem Statement}
\label{sec:background}

\paragraph{Semantic filters.}
A semantic filter $\sigma_q(R)$ keeps the tuples of a relation $R$ for which a natural-language
predicate $q$ holds, as judged by an LLM. Systems such as \lotus~\cite{patel2025lotus} and
Palimpzest~\cite{liu2024palimpzest} expose $\sigma_q$ as an operator over a language expression
(``langex'') such as ``\texttt{\{review\}} complains about the battery'', and cloud warehouses expose
it as an AI function in SQL~\cite{liskowski2026cortex}. We write $M(t,q)\in\{0,1\}$ for the label
that the \emph{oracle} model $M$ assigns to tuple $t$ under the system's prompt for $q$. Each
evaluation of $M$ is an \emph{oracle call}.

\paragraph{Reference semantics and targets.}
Following \lotus, the exact answer is the reference plan
$\sigma^*_q(R)=\{t\in R : M(t,q)=1\}$, evaluated with one tuple per prompt at temperature zero.
An approximate plan returns a set $\hat\sigma_q(R)$, judged by its precision
$|\hat\sigma\cap\sigma^*|/|\hat\sigma|$ and its recall $|\hat\sigma\cap\sigma^*|/|\sigma^*|$
against the reference. The user supplies a precision target $\gamma_P$, a recall target
$\gamma_R$, and a failure probability $\delta$. A plan \emph{meets its targets} if both measures
reach them, and it is \emph{valid} if it meets its targets with probability at least $1-\delta$
over its own randomness, for every input. Per-query optimizers such as the \lotus cascade and
BARGAIN~\cite{zeighami2025bargain} aim at such targets by scoring every tuple with a cheap proxy
and sending only uncertain tuples to the oracle; only methods with finite-sample bounds, such as
BARGAIN, are valid in this sense.

\paragraph{Workloads.}
We model a \emph{workload} as a sequence of semantic filters $q_1,q_2,\dots$ over the same
relation and columns, each with its own targets, as in Example~\ref{ex:running}. The predicates
differ, so an exact-match cache keyed by predicate and
tuple~\cite{dorbani2025flockmtl,kumarasinghe2026ipdb} never hits. Yet they are often related: many tuples
that pass one predicate also pass another, and many that fail one fail another.

\begin{problem}[Certified cross-query reuse]
\label{prob:reuse}
\textbf{Input:} a relation $R$; the oracle labels paid for by the earlier filters
$q_1,\dots,q_{j-1}$ of a workload, organized as a catalog $\mathcal{C}$; a new filter $q_j$ with
targets $(\gamma_P,\gamma_R,\delta)$.
\textbf{Output:} a set $\hat\sigma_{q_j}(R)\subseteq R$.
\textbf{Constraint:} for every relation, every workload, and every catalog, the output meets both
targets with probability at least $1-\delta$.
\textbf{Objective:} minimize the expected number of new oracle calls for $q_j$.
\end{problem}

\noindent
Two features of the problem rule out the obvious solutions. First, validity must hold for every
catalog, including one in which every earlier result is useless or misleading for $q_j$; reusing a
past answer because its predicate looks similar is not valid (Section~\ref{sec:eval-validity}).
Second, on unrelated workloads the method should cost little more than the best per-query plan,
so that enabling reuse is never a gamble. The probability in the constraint is over the method's
own sampling for $q_j$, given the catalog. Table~\ref{tab:notation} summarizes our notation.

\begin{table}[t]
\caption{Notation.}
\label{tab:notation}
\small
\begin{tabular}{@{}ll@{}}
\toprule
Symbol & Meaning \\
\midrule
$R$, $t$ & relation, tuple \\
$q$, $p$ & new predicate, earlier predicate \\
$M(t,q)$ & oracle label of $t$ for $q$ \\
$\gamma_P,\gamma_R,\delta$ & precision and recall targets, failure probability \\
$V_p$, $\mathcal{C}$ & semantic view of $p$, catalog of views \\
$p\contained{\eps}q$ & $p$ is $\eps$-contained in $q$ (Definition~\ref{def:containment}) \\
$K$, $U$ & tuples with known / unknown oracle label for $q$ \\
$\hat\pi(t)$ & score: predicted probability that $M(t,q)=1$ \\
$X$, $O$, $A$ & reject, oracle, and accept regions of $U$ \\
$n_X$, $n_A$ & certification draws from $X$ and $A$ \\
$c_X$, $c_A$ & true numbers of positives in $X$ and $A$ \\
$u_X$, $\ell_A$ & upper bound on $c_X$, lower bound on $c_A$ \\
$E$ & event that both bounds hold \\
\bottomrule
\end{tabular}
\end{table}

%% file: sections/views.tex
\section{Semantic Views and Statistical Containment}
\label{sec:views}

\paragraph{Semantic views.}
A \emph{semantic view} $V_p$ is a label column over $R$ bound to a predicate $p$, the columns its
langex reads, the oracle model, and the prompt version. Each entry of $V_p$ is $1$, $0$, or
\emph{unknown}, so views may be partial. Each known entry carries its provenance: \emph{oracle}
if it was paid for with a call to $M(t,p)$, or \emph{derived} if it was decided without a call, by
accepting or rejecting a whole region under a certificate (Section~\ref{sec:reuse}). Every query
answered by \sys leaves its view behind at no extra cost, because its oracle labels and its
region decisions are exactly the view. SemWeave~\cite{mahmood2026semweave} already speaks of
treating a prior filter result as a materialized semantic view; we make the notion precise, add
provenance and partiality, and, above all, never assume that a view is correct for another
predicate.

\paragraph{Containment is statistical.}
Classical view-based query answering~\cite{halevy2001views,goldstein2001matviews} and semantic
caching~\cite{dar1996semantic} rely on logical containment: if $p$ implies $q$, every answer to
$p$ answers $q$. Semantic predicates do not behave this way. Whether ``says the battery overheats''
implies ``complains about the battery'' is a property of the oracle's labels on this relation,
not of the two sentences, and the oracle is not logically consistent across prompts.

\begin{definition}[$\eps$-containment]
\label{def:containment}
Predicate $p$ is $\eps$-contained in $q$ on $R$, written $p \contained{\eps} q$, if
\[ \Pr_{t \sim R}\left[ M(t,q)=0 \mid M(t,p)=1 \right] \le \eps. \]
\end{definition}

\noindent
Exclusion is the dual, $\Pr[M(t,q)=1 \mid M(t,p)=1]\le\eps$, and equivalence is containment in
both directions. Relations that an LLM judge proposes from the predicate text often fail as
statistical ones: \claimEOneJudgedImpliesBadShare{} of the implications the judge asserts on our
four predicate pools have $\eps>0.1$ on the oracle's labels (Section~\ref{sec:eval-relations}).

\paragraph{Containment does not chain.}
One might hope to infer relations along a lattice of predicates, as view-matching systems do for
logical containment. The following proposition shows that this is unsafe: the error of a chain
grows with the ratio of the predicates' selectivities.

\begin{proposition}[Chained containment]
\label{lem:chain}
If $p \contained{\eps_1} q$ and $q \contained{\eps_2} r$ on $R$, then $p \contained{\eps'} r$ with
\[ \eps' = \eps_1 + \eps_2 \cdot \frac{\sel(q)}{\sel(p)}, \]
where $\sel(x)=\Pr_{t\sim R}[M(t,x)=1]$. The bound is tight whenever
$\eps_2\,\sel(q) \le (1-\eps_1)\,\sel(p)$.
\end{proposition}

\begin{proof}[Proof sketch]
Split the event $\{r=0\}$ on $q$:
\begin{align*}
\Pr[r{=}0 \mid p{=}1] &\le \Pr[q{=}0 \mid p{=}1] + \Pr[q{=}1, r{=}0, p{=}1]/\Pr[p{=}1]\\
&\le \eps_1 + \Pr[q{=}1,r{=}0]/\sel(p) = \eps_1 + \eps_2\,\sel(q)/\sel(p).
\end{align*}
For tightness, place the $\eps_1$ share of $p$'s positives outside $q$ and make them fail $r$, and
place all of $q$'s $r$-failures inside $p\wedge q$, which has room for them under the stated
condition. Appendix~\ref{app:chain} gives the full proof.
\end{proof}

\noindent
When $p$ is narrow and $q$ broad, $\sel(q)/\sel(p)$ is large, and a small error on the second link
becomes a large one. For instance, with $\sel(q)/\sel(p)=10$ and $\eps_1=\eps_2=0.02$, the chain
guarantees only $\eps'=0.22$. Relations between predicates must therefore be established for each
pair on the data, never inferred transitively. \sys goes one step further and never relies on any
relation for correctness: relations decide where to look, and fresh oracle samples decide what is
true.

%% file: sections/reuse.tex
\section{Certified Reuse}
\label{sec:reuse}

\sys answers a filter $\sigma_q$ with the operator \cf (Algorithm~\ref{alg:reuse}). The operator
obeys one rule: \emph{only fresh oracle labels, drawn uniformly at random from a region fixed
before the draws, enter a confidence bound}. Views, helper scores, and embeddings decide how the
relation is ordered and split; they never decide what is true. The operator runs in five steps:
it matches candidate views (Section~\ref{sec:matching}), fits a score on a shaping sample and
refines it where it is uncertain (Sections~\ref{sec:shaping} and~\ref{sec:refit}), plans three
regions (Section~\ref{sec:planning}), certifies them
(Section~\ref{sec:certification}), and closes any remaining gap
(Section~\ref{sec:gap}). Let $K$ be the tuples whose oracle label for $q$ is already known (from
an earlier run of $q$) and $U=R\setminus K$. Labels in $K$ are used exactly and never as samples.

\begin{algorithm}[t]
\caption{\cf$(q, R, \mathcal{C}, \gamma_P, \gamma_R, \delta)$}
\label{alg:reuse}
\small
\begin{algorithmic}[1]
\State $K \gets$ tuples with oracle labels for $q$ in $\mathcal{C}$; $U \gets R\setminus K$
\State $\mathcal{V} \gets$ the $k$ views most similar to $q$ \Comment{match}
\State label a uniform sample $S\subset U$; $K\gets K\cup S$; $U\gets U\setminus S$
\State fit $\hat\pi$ on $K$ from $\mathcal{V}$, helper scores, embeddings \Comment{shape}
\State sort $U$ by $\hat\pi$; pick $X,O,A,n_X,n_A,\alpha_X,\alpha_A$ (Problem~\ref{prob:plan}) \Comment{plan}
\For{$r=1,\dots,r_{\max}$} \Comment{refine}
  \State label the $b$ tuples of $O$ with $\hat\pi$ closest to $\tfrac12$; refit $\hat\pi$; re-plan
\EndFor
\State label $O$; draw $n_X$ from $X$, $n_A$ from $A$ without replacement \Comment{certify}
\State $u_X \gets U^{\mathrm{HG}}_{\alpha_X}$; $\ell_A \gets L^{\mathrm{HG}}_{\alpha_A}$
\While{a condition of Theorem~\ref{thm:valid} fails} \Comment{close the gap}
  \State label the next unknown tuples at the top of $X$ or the bottom of $A$
\EndWhile
\State register $V_q$: oracle labels as \emph{oracle}, $A$/$X$ decisions as \emph{derived}
\State \Return known positives $\cup$ ($A$ minus known negatives)
\end{algorithmic}
\end{algorithm}

\subsection{Matching candidate views}
\label{sec:matching}
A catalog may hold hundreds of views for a relation. Matching keeps the $k=10$ views whose
predicate text is most similar to that of $q$, by the cosine similarity of local sentence
embeddings, and passes them to the score model as features. Matching does not need to identify
the relation between $p$ and $q$: the score model weights each view by how well its labels predict
$q$'s labels on the shaping sample. This subsumes label co-occurrence and lets a view that is
negatively related to $q$ (an exclusion) be as useful as a containing one. Matching errors cost
calls, never correctness, because matching only shapes the order. \sys records no relation
judgments; an LLM judge is used in our evaluation only to drive the SemWeave baseline and to
measure how often judged relations fail.

\subsection{Shaping: views as features}
\label{sec:shaping}
\sys first labels a \emph{shaping sample} of $s_0$ tuples drawn uniformly from $U$ ($s_0=200$ by
default, at most 8\% of $U$). These labels are returned exactly and join $K$, but they never enter
a bound. On all labeled tuples, \sys fits an $L_2$-regularized logistic model $\hat\pi(t)$ whose
features cost no oracle calls:
\begin{itemize}
\item for each candidate view $V_p$, an indicator that $V_p$ labels $t$ positive and an
indicator that $V_p$ does not know $t$;
\item the helper model's score for $q$, when a cascade helper is configured, as in \lotus;
\item a 32-dimensional projection of a local sentence embedding of $t$.
\end{itemize}
Treating ``unknown'' as a feature value lets partial views contribute like complete ones. The view
indicators are what distinguishes \sys from a per-query proxy: past answers to related predicates
can be far more predictive of $q$ than a small helper model.

\subsection{Refining the score}
\label{sec:refit}
A score fitted on 200 uniform labels is sample-starved: at 5\% selectivity the shaping sample
holds about ten positives, too few to learn how much to trust twenty view features. \cf therefore
plans once, labels the $b$ tuples of the planned oracle region whose $\hat\pi$ is closest to
$\tfrac12$, refits, and plans again, for up to $r_{\max}=8$ rounds. The batch is $b=5\%$ of the
relation; a fixed batch of a few hundred boundary tuples would distort, on a large relation, the
probabilities that the planner reads as expected counts. Refinement labels tuples that the
current plan would label anyway, uses them exactly, and ends before the regions are fixed and any certification draw is
made, so validity is unaffected (Theorem~\ref{thm:valid}).

\subsection{Planning the regions}
\label{sec:planning}
Let $u_1,\dots,u_m$ be the tuples of $U$ in ascending order of $\hat\pi$. A \emph{plan}
$(i,j,n_X,n_A)$ with $i\le j$ puts the first $i$ tuples in the reject region $X$, the next $j-i$
in the oracle region $O$, and the rest in the accept region $A$; $n_X$ and $n_A$ are the numbers
of certification draws from $X$ and $A$. The planner predicts what certification would conclude by
plugging expected counts into the bounds of Section~\ref{sec:certification}. With
$\hat c_Y=\sum_{t\in Y}\hat\pi(t)$ for a region $Y$, and $\alpha$ the per-region level, it predicts
\begin{gather*}
\hat u_X = |X|\,U^{\mathrm{CP}}_\alpha\!\left(n_X\hat c_X/|X|,\; n_X\right),\qquad
\hat\ell_A = |A|\,L^{\mathrm{CP}}_\alpha\!\left(n_A\hat c_A/|A|,\; n_A\right),
\end{gather*}
where $U^{\mathrm{CP}}_\alpha(x,n)$ and $L^{\mathrm{CP}}_\alpha(x,n)$ are the one-sided
Clopper--Pearson bounds for $x$ successes in $n$ trials (with $x$ rounded). Let $k^+$ be the
known positives of $q$. The planner predicts the true positives returned, the positives missed
(positives drawn from $X$ are found and returned), and the size of the returned set as
\begin{gather*}
\hat T=k^+ + \hat c_O + \hat\ell_A,\qquad
\hat\mu=\max\bigl(\hat u_X - n_X\hat c_X/|X|,\,0\bigr),\\
\text{and}\quad k^+ + \hat c_O + |A|.
\end{gather*}

The planner is deliberately pessimistic: it predicts the number of positives among the draws as
its expectation shifted by half a standard deviation in the unfavorable direction, which trades a
few planned calls for fewer gap-closing rounds. It also chooses how to split $\delta$ between the
two regions, $\alpha_X=f\delta$ and $\alpha_A=(1-f)\delta$ with $f\in\{0.5,0.7,0.9\}$, because
certifying recall on a large reject region usually needs the larger share.

\begin{problem}[Region planning]
\label{prob:plan}
Minimize the predicted oracle calls $|O| + n_X + n_A$ over plans $(i,j,n_X,n_A)$, subject to
$\hat T \ge \gamma_R(\hat T+\hat\mu)$ and $\hat T \ge \gamma_P(k^+ + \hat c_O + |A|)$.
\end{problem}

\noindent
The planner searches boundaries $i,j$ on a grid of $g{+}1$ evenly spaced positions in $[0,m]$
($g=40$) and sample sizes on a doubling grid from 100 to 1{,}600, using only sizes below 70\% of
the region. Clopper--Pearson bounds are tabulated once per sample size, so each candidate plan is
evaluated in constant time and the search costs $O(g^2\cdot 5^2\cdot 3)$ evaluations (the last factor for the $\delta$ split), plus
$O(m\log m)$ for sorting. If no plan is predicted to cost fewer calls than $|U|$, everything is
$O$ and the operator degenerates to exact evaluation. The plan depends on the catalog, the
helper, the embeddings, and the shaping labels, but on no label drawn later.

\subsection{Certification}
\label{sec:certification}
\sys labels $O$ in full. It then draws $n_X$ tuples from $X$ and $n_A$ from $A$ uniformly
\emph{without replacement}, so the number of positives among the draws from a region is
hypergeometric with the region's (unknown, fixed) number of positives as its parameter. One-sided
exact hypergeometric bounds at levels $\alpha_X$ and $\alpha_A$ with $\alpha_X+\alpha_A=\delta$ (or
$\delta$ for the only non-empty region) are then valid for the finite region without any
approximation: $U^{\mathrm{HG}}$ is the largest count $c$ for which observing at most the drawn
positives has probability above $\alpha_X$, and $L^{\mathrm{HG}}$ is defined symmetrically. They
are markedly tighter than Clopper--Pearson bounds~\cite{clopper1934binomial} on the same draws when
the sample is a sizable share of the region. While planning, \sys approximates them by Clopper--Pearson
bounds with a finite-population correction. Let $u_X$ be the resulting upper
bound on the number of positives in $X$, $\ell_A$ the lower bound on the number of positives in
$A$, and $E$ the event that both bounds hold; by a union bound, $\Pr[E]\ge1-\delta$.

\sys returns every known positive and every tuple of $A$ except known negatives. To state when
this answer is certified, consider any state of the algorithm after certification and define:
\begin{itemize}
\item $P$: known positives outside $A$ (in $K$, in $O$, and those found in $X$);
\item $d_X$ and $\nu_X$: known positives and still-unknown tuples in $X$;
\item $d_A$ and $e_A$: known positives and known negatives in $A$;
\item $T = P + \max(\ell_A, d_A)$, a lower bound on the true positives returned;
\item $\mu = \min\bigl(\max(u_X - d_X, 0),\; \nu_X\bigr)$, an upper bound on the positives missed.
\end{itemize}

\begin{theorem}[Validity]
\label{thm:valid}
On the event $E$, if $T \ge \gamma_R(T + \mu)$ and $T \ge \gamma_P(P + |A| - e_A)$, the returned
set has recall at least $\gamma_R$ and precision at least $\gamma_P$. Hence \cf meets both
targets with probability at least $1-\delta$, for every relation, workload, and catalog.
\end{theorem}

\begin{proof}[Proof sketch]
Let $c_X$ and $c_A$ be the true numbers of positives in $X$ and $A$. The returned set contains
$P+c_A$ true positives (no positive of $A$ is dropped), has size $P+|A|-e_A$, and misses
$c_X-d_X\le\nu_X$ positives. Recall $(P+c_A)/(P+c_A+c_X-d_X)$ increases in $c_A$ and decreases in
$c_X$; precision $(P+c_A)/(P+|A|-e_A)$ increases in $c_A$. On $E$, $c_A\ge\ell_A$ and
$c_X\le u_X$, and trivially $c_A\ge d_A$, so $T$ and $\mu$ yield lower bounds on both ratios. The
argument holds simultaneously for every state in which the algorithm can stop, so data-dependent
stopping is covered. Views enter only through $\hat\pi$, which fixes $X$ and $A$ before the
draws. Appendix~\ref{app:valid} gives the full proof and its extension to any number of strata.
\end{proof}

\subsection{Closing the gap}
\label{sec:gap}
The planner works with predictions, so a condition of Theorem~\ref{thm:valid} may fail after
certification. \sys then labels the highest-scored unknown tuples of $X$ in batches of 50 (each
positive found is returned and lowers the bound on missed positives) or the lowest-scored unknown
tuples of $A$ (each negative found is dropped), and checks again. Labeling a region completely
makes its contribution exact, so the loop always terminates with both conditions satisfied. Every
intermediate state is covered by the same event $E$, so gap closing never affects validity; it
turns a near miss of the planner into a few extra calls instead of a failed query or a fall back
to exact evaluation.

\subsection{Discussion}
\label{sec:reuse-discussion}
\paragraph{Why not trust the views?}
A view $V_p$ with $p\contained{\eps}q$ suggests accepting $p$'s positives for $q$, and
$q\contained{\eps}p$ suggests rejecting $p$'s negatives. On real predicates $\eps$ is rarely small
(Table~\ref{tab:e1}), and errors compound along chains (Proposition~\ref{lem:chain}), so a system
that reuses on the strength of an inferred relation misses its targets
(Section~\ref{sec:eval-validity}). \cf uses the same information to place the regions and pays a
few hundred fresh labels to certify them.

\paragraph{A certification floor.}
Any sampling-based guarantee pays a per-query cost that depends on the predicate's selectivity but
not on the size of the relation, and on small relations this floor, not the quality of the score,
limits the savings.

\begin{proposition}[Certification floor]
\label{prop:floor}
Let $q$ have $c=sN$ positives among $N$ tuples, and let a plan reject $|X|=xN$ tuples. With
Clopper--Pearson bounds, even if no certification draw from $X$ is positive, recall $\gamma_R$ is
certified only if
\[
n_X \;\ge\; \frac{\ln(1/\alpha)}{-\ln\!\bigl(1-\rho\bigr)} \;\approx\; \kappa\,\frac{x}{s},
\qquad \rho=\frac{(1-\gamma_R)\,s}{\gamma_R\,x},\quad \kappa=\frac{\gamma_R\ln(1/\alpha)}{1-\gamma_R},
\]
unless all but at most $(1-\gamma_R)c/\gamma_R$ tuples of $X$ are labeled. The plan then costs at
least $|O|+n_X\ge(1-x-a)N+\kappa x/s$
calls, where $a=|A|/N$.
\end{proposition}

\begin{proof}[Proof sketch]
Hypergeometric bounds lower the floor noticeably only when $n_X$ is a sizable share of $|X|$.
With no positive among $n_X$ draws, the Clopper--Pearson bound is
$u_X=|X|(1-\alpha^{1/n_X})$ and
$d_X=0$. Theorem~\ref{thm:valid} needs $\mu\le(1-\gamma_R)T/\gamma_R$, and $T\le c$, so
$|X|(1-\alpha^{1/n_X})\le(1-\gamma_R)c/\gamma_R$; solving for $n_X$ gives the bound. Positive
draws only raise $u_X$.
\end{proof}

\noindent
For $\gamma_R=0.9$ and $\alpha=\delta/2=0.025$, $\kappa\approx\claimFloorCoef{}$: a predicate with
1\% selectivity needs about \claimFloorDrawsAtOnePct{} draws whenever most tuples are rejected,
which is two thirds of a 5{,}000-tuple relation but a rounding error on a relation of millions. The
bound holds for every score, including a perfect one. Any method that certifies recall from a
uniform sample of the rejected tuples faces a floor of the same order, so per-query cascades with
guarantees pay it as well.

\paragraph{How \cf differs from a per-query cascade.}
Per-query cascades~\cite{patel2025lotus,kang2020supg,zeighami2025bargain} threshold one scalar
proxy and estimate the quality of the threshold from samples of the whole relation. \cf differs
in three ways. Its score combines free past labels of other predicates with the usual proxies.
Every label it already holds, from an earlier run of $q$, the shaping sample, the oracle region,
or gap closing, enters the precision and recall accounting exactly rather than through an
estimator. And it certifies two regions with exact bounds chosen by a planner that trades the
size of the oracle region against the number of certification draws. With views disabled, \cf is
a finite-sample stratified cascade; we report that configuration separately so that the effect of
views and the effect of the certifier can be told apart (Section~\ref{sec:eval-cost}).

%% file: sections/maintenance.tex
\section{Maintaining Semantic Views}
\label{sec:maintenance}

Views are a by-product of answering queries, so maintaining them must be cheap and must never put
the guarantee at risk. \sys meets both requirements with one principle: a view is never assumed
correct, so every change to the data or to the oracle is handled by re-certifying, not by
recomputing.

\paragraph{Registering a view.}
When \cf finishes, it registers the view of $q$: every oracle label with provenance \emph{oracle},
and the unlabeled tuples of $A$ and $X$ as derived positives and negatives. If a view of $q$
already exists, the two are merged tuple by tuple: an oracle label always wins, and a derived
label only fills an entry that has no oracle label. Only oracle labels of $q$ itself are used
exactly when $q$ is asked again; derived labels, like views of other predicates, only shape the
score. A repeated predicate is therefore answered with the labels it already paid for, which
subsumes an exact-match cache without trusting any derived decision.

\paragraph{Appends.}
Rows are identified by a content hash of the columns that the predicate reads
(Section~\ref{sec:system}). A new row enters every view as \emph{unknown}, which is already a
feature value of the score model. The view remains usable, and the score model learns from the
shaping sample how to order new rows. No view is invalidated and no label is recomputed. We
measure this case in Section~\ref{sec:eval-robust}.

\paragraph{Oracle upgrades.}
A view's identity includes the oracle model and the prompt version. After an upgrade, queries
under the new oracle find no oracle labels of their own, so nothing from the old oracle is used
exactly. Views recorded under the old oracle remain usable as features, because
Theorem~\ref{thm:valid} never assumes that features are correct. Section~\ref{sec:eval-robust}
shows that such views still save calls.

\paragraph{Storage and lifecycle.}
In memory, a view costs two bytes per tuple: one for the label and one for the provenance. On disk,
the catalog stores only known entries. All our experiments keep every view. Two lifecycle
questions remain open: which partial views to complete during idle time, and which views to drop
under a storage budget. Both are selection problems over the workload history, in the spirit of
classical view selection~\cite{harinarayan1996cubes}; we leave a budgeted advisor to future work.

%% file: sections/system.tex
\section{System}
\label{sec:system}

\sys is a Python package with an engine-independent core and a thin adapter for
\lotus~\cite{patel2025lotus,lotus2026software} (Figure~\ref{fig:arch}). The core contains the
catalog, view matching, the certification bounds, and \cf. It does not import \lotus, and a test
enforces this separation, so adapters for other engines need only supply an oracle, optional
helper scores, and row identities.

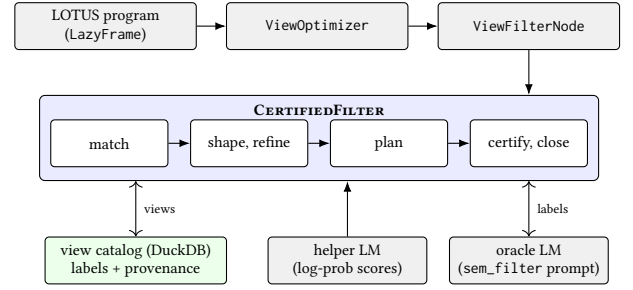
\begin{figure}[t]
\centering
\begin{tikzpicture}[font=\scriptsize,
  box/.style={draw,rounded corners=1.5pt,align=center,inner sep=2pt,minimum height=0.62cm},
  eng/.style={box,fill=black!6}, opstep/.style={box,fill=white,minimum width=1.55cm},
  arr/.style={-{Latex[length=1.4mm]}}]
  \node[eng,minimum width=2.4cm] (prog) at (0,0) {\lotus program\\(\texttt{LazyFrame})};
  \node[eng,minimum width=2.4cm] (opt) at (2.8,0) {\texttt{ViewOptimizer}};
  \node[eng,minimum width=2.4cm] (node) at (5.6,0) {\texttt{ViewFilterNode}};
  \draw[arr] (prog) -- (opt);
  \draw[arr] (opt) -- (node);
  \node[opstep] (m) at (0.05,-1.55) {match};
  \node[opstep] (sh) at (1.9,-1.55) {shape, refine};
  \node[opstep] (pl) at (3.75,-1.55) {plan};
  \node[opstep] (ce) at (5.6,-1.55) {certify, close};
  \draw[arr] (m) -- (sh); \draw[arr] (sh) -- (pl); \draw[arr] (pl) -- (ce);
  \begin{scope}[on background layer]
  \node[draw,rounded corners=2pt,fill=blue!8,fit=(m)(ce),
        inner sep=4pt,inner ysep=7pt,yshift=2pt,label={[anchor=north,font=\scriptsize\bfseries,yshift=-0.5pt]north:\cf}] (cf) {};
  \end{scope}
  \draw[arr] (node) -- (node |- cf.north);
  \node[box,fill=green!8,minimum width=2.4cm] (cat) at (0.4,-3.1) {view catalog (DuckDB)\\labels + provenance};
  \node[eng,minimum width=2.1cm] (hlp) at (3.2,-3.1) {helper LM\\(log-prob scores)};
  \node[eng,minimum width=2.1cm] (orc) at (5.6,-3.1) {oracle LM\\(\texttt{sem\_filter} prompt)};
  \draw[arr,<->] (cat.north) -- (cat.north |- cf.south) node[pos=0.5,right,font=\tiny]{views};
  \draw[arr] (hlp.north) -- (hlp.north |- cf.south);
  \draw[arr,<->] (orc.north) -- (orc.north |- cf.south) node[pos=0.5,right,font=\tiny]{labels};
\end{tikzpicture}
\caption{Architecture. The \lotus adapter (gray) rewrites semantic filters into nodes that run
\cf (blue) against a persistent view catalog. Oracle calls use \lotus's own prompt and model, so
answers keep \lotus's reference semantics.}
\label{fig:arch}
\end{figure}

\paragraph{\lotus integration.}
The adapter is a \lotus \texttt{LazyFrame} optimizer. It rewrites each \texttt{SemFilterNode}
into a \texttt{ViewFilterNode} that keeps every field of the original node and is bound to a
catalog; every other node, prompt, and model is left unchanged, and no \lotus source is modified.
Users opt in by adding one optimizer to a program:
\begin{lstlisting}
catalog = Catalog("reviews", 0, path="reviews.views.duckdb")
args = CascadeArgs(recall_target=0.9, precision_target=0.9,
                   failure_probability=0.05)
lf = LazyFrame().sem_filter("{text} says the battery overheats",
                            cascade_args=args)
out = lf.optimize([ViewOptimizer(catalog)]).execute(df)
\end{lstlisting}
The targets are read from \lotus's own \texttt{CascadeArgs}. The oracle is \lotus's own
\texttt{sem\_filter} prompt on the configured model, and helper scores are \lotus's own
log-probability scores from its helper model, so \sys answers with respect to exactly the
reference semantics that \lotus defines. Filters with options that \sys does not support
(explanations, raw outputs, few-shot examples, or a non-default reasoning strategy) run through
the native \lotus path.

\paragraph{Identities and catalog.}
A row's identity is a content hash of the columns that the langex reads, so views survive across
DataFrames and sessions, and rows of a new DataFrame that the catalog has not seen are appended to
every view as unknown. A view's identity is a hash of the langex, the prompt version, and the
oracle model. The catalog is a DuckDB file beside the data with four tables: the table's
metadata, the row identities, one row per view (predicate text, model, and the targets under
which it was produced), and the known labels with their provenance, stored sparsely.

\paragraph{Cost of the operator itself.}
Besides oracle calls, \cf embeds the tuples locally and, per refinement round, fits a small
logistic model and runs the planner's grid search: a median of
\claimEThreeSemviewsMsPerQuery{}\,ms per query in our single-threaded implementation.

%% file: sections/evaluation.tex
\section{Evaluation}
\label{sec:eval}

We answer six questions, one per subsection. (1) Do real workloads contain the relations that
reuse needs, and can relations judged by an LLM be trusted? (2) How many oracle calls does \sys
save, and where do the savings come from? (3) Do the guarantees hold, and do uncertified
alternatives keep them? (4) Is reuse safe and profitable when a question returns in other words?
(5) How robust is \sys to unrelated workloads, adversarial or partial views, appends, and oracle
changes? (6) How sensitive is it to its parameters and components?

\subsection{Experimental Setup}
\label{sec:eval-setup}

\paragraph{Datasets and workloads.}
Table~\ref{tab:datasets} summarizes the four text relations, each a sample of 5{,}000 tuples drawn
with a fixed seed. Measuring validity requires the oracle's label for every tuple and predicate, so
we limit the relations' size to limit the cost and carbon footprint of the experiments; larger-scale
experiments (Appendix~\ref{app:scale}) confirm that the findings hold on larger relations. \emph{PubMed} pairs biomedical abstracts with the 25 human-written predicates
that ScaleDoc~\cite{zhang2025scaledoc} released for this corpus, used verbatim; it is our external
workload. \emph{GoEmotions} and \emph{DBpedia} (social-media comments and encyclopedia entries)
each have 40 predicates generated mechanically from the datasets' published label taxonomies
(emotions, Ekman groups, and sentiments; three levels of DBpedia classes), so their logical
containments are known by construction. \emph{Reviews} has 40 analyst-style predicates over
electronics product reviews (battery, defects, shipping, service, sentiment, \dots), written
before any label was seen. A workload issues all predicates of a pool in a random order. For each
pool we also build a control workload of eight nearly uncorrelated predicates. Finally, a
\emph{reworded} workload models users and agents that ask the same question again in other words:
for \claimWrTwins{} predicates drawn at random from each of the Reviews and GoEmotions pools, gpt-oss-120b
wrote one rewording that keeps the meaning, which the oracle labeled like every other
predicate; the workload issues each pool together with its rewordings in random order.

\begin{table}[t]
\caption{Datasets and predicate pools. Selectivity is the median share of tuples that the oracle
labels positive, with the range over predicates in parentheses.}
\label{tab:datasets}
\small
\begin{adjustbox}{max width=\columnwidth}\input{tables/datasets}\end{adjustbox}
\end{table}

\begin{figure*}[t]
\centering
\includegraphics[width=\textwidth]{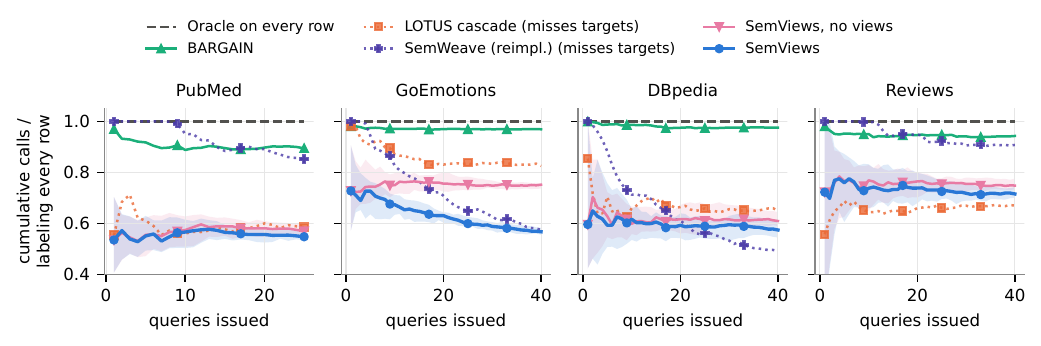}
\caption{Cumulative oracle calls as a share of exact evaluation as each workload unfolds (median
over workload orders; bands: 10th to 90th percentile). Dotted methods miss their targets on many
queries (Table~\ref{tab:e3}); \sys is the cheapest method that meets them.}
\label{fig:e2}
\end{figure*}

\paragraph{Oracle, helper, and recorded labels.}
The oracle is Llama~3.3~70B Instruct at temperature zero, prompted with \lotus's own
\texttt{sem\_filter} prompt, one tuple per prompt. The helper is Llama~3.2~3B Instruct (4-bit)
served locally; its score is \lotus's own log-probability score. We recorded every
(tuple, predicate) oracle label and helper score once, \claimTapeOracleCalls{} oracle calls in all,
and replay every method against this record. Oracle-call counts are therefore exact, and validity
can be checked over hundreds of workload orders. Re-labeling 2{,}000 random pairs twice reproduced
the recorded label in \claimNoiseAgreement{} of cases. Because the oracle is not perfectly
deterministic, the reference answer, and hence every precision, recall, and validity number, is
defined with respect to the recorded labels, as in \lotus and Larch~\cite{zhao2026larch}. A second
oracle, gpt-oss-120b, labels the first 1{,}000 GoEmotions tuples for the oracle-change experiment.
Row and predicate embeddings come from a local sentence encoder.

\paragraph{Methods.}
\emph{Exact evaluation} labels every tuple. The per-query methods with targets are:
\begin{itemize}
\item \emph{\lotus cascade}: \lotus~1.2.4's own importance sampling, calibration, and threshold
code on the recorded helper scores, with its default \texttt{CascadeArgs} (a 10\% importance
sample drawn from the first 200 tuples, 50 calibration quantiles); and a \emph{full-range}
variant whose importance sample spans the whole relation;
\item \emph{BARGAIN}~\cite{zeighami2025bargain}: its released code, with joint precision and
recall targets, on the helper score;
\item \emph{embedding proxy}: logistic regression on row embeddings trained on a 5\% labeled
sample and thresholded with the \lotus cascade, following~\cite{chung2026proxy};
\item \emph{\sys, no views}: \cf with only helper and embedding features.
\end{itemize}
The cross-query methods are:
\begin{itemize}
\item \emph{exact cache}: no predicate repeats, so it equals its \lotus-cascade fallback and is
omitted from the tables;
\item \emph{similarity cache}: reuses the labels of the most similar past predicate when the cosine
similarity of the predicate texts exceeds $\tau\in\{0.7,0.8,0.9\}$, and otherwise runs the \lotus
cascade; a \emph{similarity proxy} instead feeds those labels to the \lotus cascade as proxy
scores;
\item \emph{SemWeave}~\cite{mahmood2026semweave}, reimplemented from its Algorithm~1: for each new
filter it retrieves the five most similar past filters, obtains their bidirectional entailment
relation from an LLM judge (one prompt per pair, with SemWeave's consistency check), and decides a
tuple without the oracle whenever a stored result and the relation imply its label;
\item \emph{\sys}: \cf with views.
\end{itemize}
SemWeave has no precision or recall targets; we report how often its answers would meet ours.
Unless stated otherwise, $\gamma_P=\gamma_R=0.9$ and $\delta=0.05$.

\paragraph{Metrics and protocol.}
We report oracle calls and the share of queries that miss either target. Costs are medians over
20 workload orders. Validity is measured over \claimEThreeSemviewsSeeds{} orders for \sys
(\claimEThreeSemviewsQueries{} queries) and 200 orders for the baselines. We count calls rather
than dollars because the two agree: every recorded call carries its token counts, and for every
method and dataset, the share of exact evaluation's dollar cost differs from the share of its
calls by at most \claimETwoUsdCallGap{} percentage points.

\subsection{Relations Among Real Predicates}
\label{sec:eval-relations}

\paragraph{Containment is real but rarely tight, and judged implications usually fail.}
Table~\ref{tab:e1} measures the containment error $\eps=\Pr[q{=}0\mid p{=}1]$ on the oracle's
labels for every ordered pair of predicates over the same relation. Across the four pools, only
\claimEOneContainZeroOnePairs{} of \claimEOnePairs{} ordered pairs are $0.1$-contained, while
near-exclusion is common. The LLM judge called \claimEOneJudgedImplies{} pairs implications, and
\claimEOneJudgedImpliesBadShare{} of those have $\eps>0.1$. Implications derived from a taxonomy
fare no better, because the oracle is not logically consistent across prompts. This is the
failure that uncertified reuse inherits.

\begin{table}[t]
\caption{Containment and exclusion among predicates on the oracle's labels: ordered pairs that are
$\eps$-contained ($\sqsubseteq_\eps$) or $\eps$-exclusive (Excl.$_\eps$), and implications asserted
by an LLM judge, with the share that fail statistically ($\eps>0.1$).}
\label{tab:e1}
\small
\begin{adjustbox}{max width=\columnwidth}\input{tables/e1_relations}\end{adjustbox}
\end{table}

\subsection{Oracle Calls}
\label{sec:eval-cost}

\paragraph{\sys is the cheapest valid method on every dataset.}
Figure~\ref{fig:e2} plots cumulative oracle calls as each workload unfolds. \sys saves
\claimETwoSemviewsSavingMin{}--\claimETwoSemviewsSavingMax{} of the calls of exact evaluation.
BARGAIN with the helper score, the only prior per-query method that meets its targets on every
dataset, saves \claimETwoBargainSavingMin{}--\claimETwoBargainSavingMax{}: \sys saves
\claimETwoSavingRatioMin{}--\claimETwoSavingRatioMax{} as much, and BARGAIN needs
\claimETwoVsBargainMin{}--\claimETwoVsBargainMax{} the calls of \sys. Against the cheapest
per-query method that meets its targets on at least 95\% of queries (in every case \sys without
views), the factor is \claimETwoRatioValidMin{}--\claimETwoRatioValidMax{}. The \lotus cascade and
SemWeave appear cheaper on some datasets, but they do not meet their targets
(Section~\ref{sec:eval-validity}).

\paragraph{Later queries.}
The cumulative totals include the first queries of each workload, which find an empty catalog.
Over the last ten queries, \sys needs as little as \claimETwoLateShareMin{} of the calls of exact
evaluation, and BARGAIN needs \claimETwoLateVsBargainMin{}--\claimETwoLateVsBargainMax{} the calls
of \sys. Views remove up to \claimETwoLateViewGainMax{} of the calls of the operator without them
on GoEmotions, the pool richest in related predicates; on the other pools, their late gain stays
close to their workload average (Table~\ref{tab:views}).

\begin{table}[t]
\caption{Calls saved by views, relative to \sys without views on the same workload order.
\emph{Related preds}: predicates whose oracle labels correlate with another predicate of the pool
($|\phi|>0.6$). \emph{Related} and \emph{unrelated}: queries for which an earlier query of the
order did or did not ask such a predicate. \emph{All}: the whole workload; \emph{last 10}: its last
ten queries.}
\label{tab:views}
\small
\begin{adjustbox}{max width=\columnwidth}\input{tables/views}\end{adjustbox}
\end{table}

\paragraph{Views help when the workload holds related predicates.}
Against the same operator with views disabled, views remove
\claimETwoViewGainMin{}--\claimETwoViewGainMax{} of the calls over a whole workload, but this
average mixes two kinds of queries (Table~\ref{tab:views}). When an earlier query asked a strongly
correlated predicate, views remove \claimViewsRelatedGainMin{}--\claimViewsRelatedGainMax{} of
the calls; otherwise they remove at most
\claimViewsUnrelatedGainMax{}, and at most \claimViewsUnrelatedGainOtherMax{} outside GoEmotions,
where overlapping emotions also help weakly correlated predicates. The workload
average therefore follows how often a pool contains related predicates:
\claimViewsPartnersGoemotions{} of 40 GoEmotions predicates have a strongly correlated partner, and
at most \claimViewsPartnersOtherMax{} in each other pool. On DBpedia, the embedding alone already
separates most topical classes, so even related views add less. Refinement matters here: the
first-draft operator, which fits its score on the 200-tuple shaping sample alone, gained only
\claimViewsVOneRelatedGainMin{}--\claimViewsVOneRelatedGainMax{} from related views, because so few
labels cannot teach the score how much to trust twenty view features.

\paragraph{Certification, features, and refinement contribute separately.}
To isolate the certifier, we give BARGAIN the same shaping sample and the same score that \sys
fits on it before refinement. BARGAIN's threshold search then needs
\claimETwoVsBargainViewsVOneMin{}--\claimETwoVsBargainViewsVOneMax{} the calls of region certification
on that score (the first-draft operator of Table~\ref{tab:ablation}), and refining the score removes
another \claimETwoRefineCutMin{}--\claimETwoRefineCutMax{} of the calls.

\begin{figure}[t]
\centering
\includegraphics[width=\columnwidth]{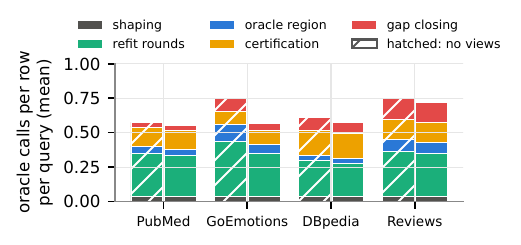}
\caption{Where the oracle calls go: mean calls per tuple and query, split by the step of \cf that
spends them, without views (hatched) and with views. Refinement and the oracle region label the
tuples the score cannot settle; views shrink both.}
\label{fig:breakdown}
\end{figure}

\paragraph{Most calls label tuples the score cannot settle.}
Figure~\ref{fig:breakdown} splits the calls of \cf by step. Refinement rounds and the oracle region,
which both label the tuples the score is unsure about, account for
\claimETwoBreakdownRefitRoundsMin{}--\claimETwoBreakdownRefitRoundsMax{} and
\claimETwoBreakdownOracleRegionMin{}--\claimETwoBreakdownOracleRegionMax{} of the calls, certification
draws for \claimETwoBreakdownCertificationMin{}--\claimETwoBreakdownCertificationMax{}, gap closing for
\claimETwoBreakdownGapClosingMin{}--\claimETwoBreakdownGapClosingMax{}, and the shaping sample for
\claimETwoBreakdownShapingMin{}--\claimETwoBreakdownShapingMax{}. Views reduce cost by sharpening the
score, so that fewer tuples need an exact label before the planner can place the rest in the
certified regions. The cost of certification itself is nearly constant, so its share
falls as relations grow (Section~\ref{sec:eval-sensitivity}).

\paragraph{Rare predicates sit at the certification floor.}
Cost depends strongly on selectivity. Predicates with at most 2\% selectivity cost
\claimSelLowShare{} of exact evaluation, those above 10\% cost \claimSelHighShare{}. This is the
floor of Proposition~\ref{prop:floor}: on 5{,}000 tuples, certifying recall for a rare predicate
needs a sample comparable to the relation itself, whatever the score. Of the
\claimSelPreds{} predicates, \claimSelLowPreds{} have at most 5\% selectivity, so the floor weighs on
every average we report, most on DBpedia and Reviews, whose median selectivities are lowest
(Table~\ref{tab:datasets}).

\subsection{Validity}
\label{sec:eval-validity}

\begin{table}[t]
\caption{Share of queries that miss a target ($\gamma_P=\gamma_R=0.9$, $\delta=0.05$). Bold:
above $\delta$. Only \sys, its view-free variant, and BARGAIN are valid on every dataset.}
\label{tab:e3}
\small
\begin{adjustbox}{max width=\columnwidth}\input{tables/e3_validity}\end{adjustbox}
\end{table}

\paragraph{\sys meets its guarantee on every dataset and every predicate.}
Table~\ref{tab:e3} reports the share of queries that miss either target. \sys misses on at most
\claimEThreeSemviewsMaxMiss{} of queries on any dataset, below $\delta=\claimEThreeDelta{}$, as
Theorem~\ref{thm:valid} requires. Because the guarantee is per query, we also check each predicate
separately: across \claimEThreeSemviewsPredicates{} predicates, the worst per-predicate miss rate
over its \claimEThreeSemviewsSeeds{} runs is \claimEThreeSemviewsWorstPredicateMiss{} (95\% upper
confidence bound \claimEThreeSemviewsWorstPredicateUpper{}), and
\claimEThreeSemviewsPredicatesAboveDelta{} predicates exceed $\delta$. BARGAIN is valid as well
(at most \claimEThreeBargainMaxMiss{}).

\paragraph{The \lotus cascade misses often, even with stricter internal targets.}
The \lotus cascade misses on \claimEThreeLotusMinMiss{}--\claimEThreeLotusMaxMiss{} of queries,
and \claimEThreeLotusPredicatesAboveDelta{} of \claimEThreeLotusPredicates{} predicates exceed
$\delta$. Its thresholds come from normal approximations on an importance sample drawn from the
first 200 tuples; the full-range variant still misses on up to
\claimEThreeLotusCascadeFullMaxMiss{}. Raising its targets by 0.08 and setting $\delta=0.001$ does
not fix this: it still misses on
\claimEOneTwoStrictLotusMissMin{}--\claimEOneTwoStrictLotusMissMax{} of queries, is valid on only
\claimEOneTwoStrictLotusValidDatasets{} of four datasets, and costs nearly as much as exact
evaluation.

\paragraph{Uncertified reuse is not safe.}
Similarity caches miss on \claimEThreeSimcacheMinMiss{}--\claimEThreeSimcacheMaxMiss{} of
queries. Part of these misses comes from their fallback, the \lotus cascade, but when the cache
actually hits, the answer misses a target \claimEThreeSimcacheMissGivenHit{} of the time. SemWeave
misses on \claimEThreeSemweaveMinMiss{}--\claimEThreeSemweaveMaxMiss{} of queries at
\claimETwoSemweaveShareMin{}--\claimETwoSemweaveShareMax{} of the cost of exact evaluation. It
reuses stored results on \claimEThreeSemweaveReuseShare{} of queries and misses on
\claimEThreeSemweaveMissGivenReuse{} of those, with a 5th-percentile recall of
\claimEThreeSemweaveRecallPFive{}: entailment judged from the predicate text is not entailment on
the oracle's labels. The embedding proxy, which trusts a learned classifier without certifying
it, misses on up to \claimEThreeProxyLrMaxMiss{}.

\begin{figure}[t]
\centering
\includegraphics[width=\columnwidth]{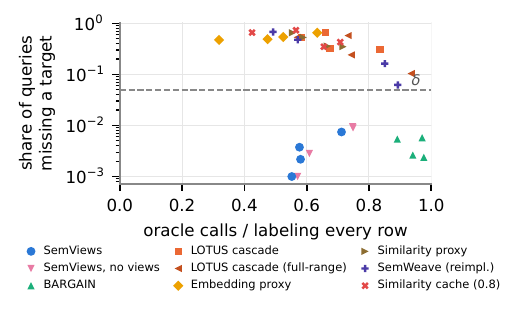}
\caption{Cost against validity. Each point is one method on one dataset. Only methods below the
dashed line meet the guarantee; among them, \sys is the cheapest.}
\label{fig:tradeoff}
\end{figure}

\paragraph{Cost and validity together.}
Figure~\ref{fig:tradeoff} combines both measures. On every dataset, each method that is cheaper
than \sys misses its targets on at least \claimEThreeCheaperThanSemviewsMinMiss{} of queries, and
each method that meets the guarantee costs more than \sys.

\begin{figure*}[t]
\centering
\includegraphics[width=\textwidth]{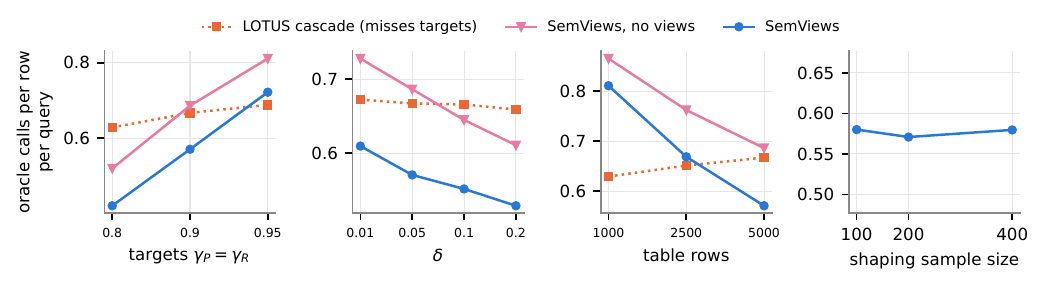}
\caption{Oracle calls per tuple and query (median over all four datasets and workload orders) as
the targets, $\delta$, the table size, and the shaping-sample size vary one at a time. Larger
tables help because the certification floor does not grow with the table.}
\label{fig:e4}
\end{figure*}

\subsection{Repeated Questions in Different Words}
\label{sec:eval-reworded}

\begin{table}[t]
\caption{Reworded workload: oracle calls (share of exact evaluation) and share of queries that miss
a target, for queries whose twin, the same question in other words, was asked earlier (calls: \claimWrOrdersCost{} workload orders; misses: \claimWrOrdersValid{}). Bold: above
$\delta$.}
\label{tab:reworded}
\small
\begin{adjustbox}{max width=\columnwidth}\input{tables/rephrased}\end{adjustbox}
\end{table}

\paragraph{The oracle does not answer a rewording the same way.}
A reworded predicate agrees with its original on \claimWrAgreeMin{}--\claimWrAgreeMax{} of the tuples,
but because most tuples are negatives, the overlap of the two positive sets (Jaccard) is only
\claimWrJaccardMin{}--\claimWrJaccardMax{}. Reusing the twin's answer therefore breaks a 90\% precision
or recall target more often than not: SemWeave, whose judge correctly calls almost every twin
equivalent, and the similarity caches miss their targets on
\claimWrUnsafeMissMin{}--\claimWrUnsafeMissMax{} of these queries (Table~\ref{tab:reworded}).

\paragraph{Certified reuse turns the twin into savings and keeps the guarantee.}
\sys uses the twin's view only to order the tuples. It needs
\claimWrSemviewsShareMin{}--\claimWrSemviewsShareMax{} of the calls of exact evaluation on repeated
questions, \claimWrViewGainMin{}--\claimWrViewGainMax{} fewer than without views, and misses its targets
on at most \claimWrSemviewsMissMax{} of them over \claimWrOrdersValid{} workload orders.

\subsection{Robustness}
\label{sec:eval-robust}

\paragraph{Unrelated workloads cost nothing measurable.}
On the control workloads of eight nearly uncorrelated predicates, views add no measurable cost
(Table~\ref{tab:e9}; worst case +\claimENineWorstOverhead{}) and reduce calls slightly on three
of the four pools. The shaping sample is
needed anyway, and an uninformative view only adds features that the score model learns to
ignore.

\begin{table}[t]
\caption{Oracle calls on control workloads of nearly uncorrelated predicates (median over
workload orders). Views cost almost nothing when they cannot help.}
\label{tab:e9}
\small
\begin{adjustbox}{max width=\columnwidth}\input{tables/e9_overhead}\end{adjustbox}
\end{table}

\paragraph{Adversarial and inverted views do not break the guarantee.}
We seed the catalog with a misleading view for every predicate: the same predicate text, so
matching ranks it first, but labels drawn at random with the predicate's selectivity. \sys still
meets its targets on at least \claimEOneZeroAdversarialMetMin{} of queries and costs no more than
the view-free operator (cost ratio at most \claimEOneZeroAdversarialCostRatioMax{}, because the
workload's genuine views still accumulate). A similarity cache that trusts the matching text meets
its targets on \claimEOneZeroSimcacheMet{} of queries. A view whose labels are the \emph{inverse}
of the predicate's (with 20\% noise) is not misleading at all for \sys: the score model learns a
negative weight, and calls drop by \claimEOneZeroInvertedSavingMin{}--\claimEOneZeroInvertedSavingMax{}.

\paragraph{Appends.}
We record views for every predicate on the first 80\% of the tuples, append the remaining 20\%,
and run the workload again. The new tuples are unknown in every view. \sys meets its targets on
\claimEOneOneAppendMet{} of queries and needs \claimEOneOneAppendShareMin{}--\claimEOneOneAppendShareMax{}
of the calls of exact evaluation, far less than in a fresh workload because every predicate
already has a view. A similarity cache that reuses the stale labels meets its targets on
\claimEOneOneSimcacheMet{} of queries.

\paragraph{Views survive an oracle change.}
We answer GoEmotions queries under a second oracle, gpt-oss-120b, with views recorded under
Llama~3.3 (Table~\ref{tab:e7}). \sys meets its targets on \claimESevenCrossModelMet{} of queries
and saves \claimESevenCrossModelSaving{} of the calls of the view-free operator, more than with
views from the same oracle, because all views exist from the first query. A similarity cache that
reuses the old oracle's labels for the same predicate meets its targets on
\claimESevenSimcacheCrossMet{} of queries.

\begin{table}[t]
\caption{Views recorded under one oracle and reused under another (GoEmotions, 1{,}000 tuples):
total oracle calls and share of queries that meet both targets.}
\label{tab:e7}
\small
\begin{adjustbox}{max width=\columnwidth}\input{tables/e7_robustness}\end{adjustbox}
\end{table}

\paragraph{Partial views.}
Figure~\ref{fig:e5} keeps a random share of each view's labels. On GoEmotions, where views help
most, savings grow steadily with coverage; half of each view still saves
\claimEFiveGoemotionsHalfSaving{} of the calls of the view-free operator. On the other datasets,
partial views keep the cost within \claimEFivePartialOtherMaxDev{} of the view-free operator,
sometimes slightly above it, because sparse views add features that carry little signal.

\begin{figure}[t]
\centering
\includegraphics[width=\columnwidth]{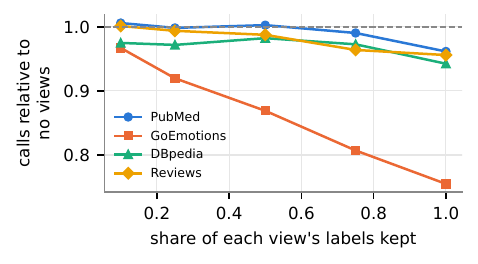}
\caption{Oracle calls relative to the view-free operator as each view keeps only a random share
of its labels. Savings grow with coverage where views are informative; sparse views cost little.}
\label{fig:e5}
\end{figure}

\subsection{Sensitivity, Ablations, and Overheads}
\label{sec:eval-sensitivity}

\paragraph{Targets, $\delta$, and table size.}
Figure~\ref{fig:e4} varies one parameter at a time. Costs rise with stricter targets and fall
slowly with larger $\delta$, since the bounds depend on $\log(1/\delta)$. Larger tables help:
the certification floor does not grow with the table (Proposition~\ref{prop:floor}), so its share
of the calls falls as tables grow, and the operator's cost falls from \claimEFourRowsSmallShare{} of exact evaluation at 1{,}000 tuples to
\claimEFourRowsLargeShare{} at 5{,}000 (median over datasets and orders).

\paragraph{Ablations.}
Table~\ref{tab:ablation} removes or replaces one component at a time, pairing each variant with
\sys on the same workload order. Bounding each view-signature stratum separately (the
per-signature strata variant) costs substantially more: real strata are not pure enough, and
per-stratum bounds dilute $\delta$. Row embeddings are the most valuable free signal; without them
the operator is markedly more expensive on every dataset. The helper's score helps on PubMed and
barely matters elsewhere. Considering all views instead of the top ten barely changes the cost. With refinement, the size of the shaping sample
hardly matters: refinement rounds supply the labels that a larger sample would. The first-draft
operator, without refinement and with Clopper--Pearson bounds on draws with replacement, needs
\claimETwoRefineCutInvMin{}--\claimETwoRefineCutInvMax{} the calls of \sys.

\begin{table}[t]
\caption{Ablations: oracle calls of each variant relative to \sys on the same workload order,
median and [10th, 90th percentile] over orders (lower is better).}
\label{tab:ablation}
\small
\begin{adjustbox}{max width=\columnwidth}\input{tables/ablation}\end{adjustbox}
\end{table}

\paragraph{Latency.}
All methods replay the same recorded labels, so we report latency through its two drivers.
First, the operator's own computation (Section~\ref{sec:system}) takes a median of
\claimEThreeSemviewsMsPerQuery{}\,ms per query. Second, oracle calls are issued in sequential rounds: the shaping sample, the
refinement rounds (median \claimETwoRefitRoundsMedian{}, 90th percentile
\claimETwoRefitRoundsPNineZero{}), then the oracle region and the certification draws in parallel,
then gap closing in batches of 50. Gap closing is
not needed in \claimETwoCloseStepsZeroShare{} of queries; the median number of gap-closing rounds
is at most \claimETwoCloseStepsMedianMax{} on any dataset, and the 90th percentile at most
\claimETwoCloseStepsPNineZeroMax{}. Each round adds the latency of one batch of parallel oracle calls, the main latency cost of
\cf relative to a single-round cascade; fewer, larger rounds would trade calls for latency.

%% file: tables/datasets.tex
\begin{tabular}{@{}lllrr@{}}
\toprule
Dataset & Text & Predicate pool & Preds & Selectivity \\
\midrule
PubMed & Biomedical abstracts & Human-written, external & 25 & 13\% (0.7--51\%) \\
GoEmotions & Social-media comments & From emotion taxonomy & 40 & 23\% (3.0--57\%) \\
DBpedia & Encyclopedia entries & From class taxonomy & 40 & 5\% (0.9--59\%) \\
Reviews & Product reviews & Analyst-style, ours & 40 & 7\% (0.2--78\%) \\
\bottomrule
\end{tabular}

%% file: tables/e1_relations.tex
\begin{tabular}{lrrrrr}
\toprule
Dataset & Preds & $\sqsubseteq_{0.05}$ & $\sqsubseteq_{0.1}$ & Excl.$_{0.05}$ & Judged $\Rightarrow$ (fail) \\
\midrule
PubMed & 25 & 1 & 5 & 123 & 12 (58\%) \\
GoEmotions & 40 & 12 & 29 & 268 & 58 (64\%) \\
DBpedia & 40 & 29 & 37 & 1107 & 62 (45\%) \\
Reviews & 40 & 11 & 21 & 672 & 14 (71\%) \\
\bottomrule
\end{tabular}

%% file: tables/views.tex
\begin{tabular}{@{}lrrrrr@{}}
\toprule
& Related & \multicolumn{4}{c}{Calls saved by views} \\
\cmidrule(l){3-6}
Dataset & preds & related & unrelated & all & last 10 \\
\midrule
PubMed & 4/25 & 11\% & 3\% & 4\% & 5\% \\
GoEmotions & 25/40 & 38\% & 14\% & 25\% & 31\% \\
DBpedia & 9/40 & 19\% & 5\% & 6\% & 2\% \\
Reviews & 9/40 & 30\% & 1\% & 4\% & 4\% \\
\bottomrule
\end{tabular}

%% file: tables/e3_validity.tex
\begin{tabular}{lrrrr}
\toprule
Method & PubMed & GoEmotions & DBpedia & Reviews \\
\midrule
SemViews & 0.0\% & 0.4\% & 0.2\% & 0.7\% \\
SemViews, no views & 0.0\% & 0.9\% & 0.3\% & 0.9\% \\
BARGAIN & 0.5\% & 0.6\% & 0.2\% & 0.3\% \\
LOTUS cascade & \textbf{53.5\%} & \textbf{30.8\%} & \textbf{65.4\%} & \textbf{32.4\%} \\
LOTUS cascade (full-range) & \textbf{53.1\%} & \textbf{10.4\%} & \textbf{57.5\%} & \textbf{24.1\%} \\
Embedding proxy & \textbf{49.0\%} & \textbf{65.4\%} & \textbf{47.3\%} & \textbf{54.2\%} \\
Similarity proxy & \textbf{52.9\%} & \textbf{35.0\%} & \textbf{65.3\%} & \textbf{35.8\%} \\
SemWeave (reimpl.) & \textbf{16.3\%} & \textbf{47.6\%} & \textbf{67.8\%} & \textbf{6.2\%} \\
Similarity cache (0.7) & \textbf{73.0\%} & \textbf{60.5\%} & \textbf{80.3\%} & \textbf{36.4\%} \\
Similarity cache (0.8) & \textbf{66.0\%} & \textbf{42.9\%} & \textbf{73.3\%} & \textbf{34.9\%} \\
Similarity cache (0.9) & \textbf{59.1\%} & \textbf{30.8\%} & \textbf{67.7\%} & \textbf{32.4\%} \\
\bottomrule
\end{tabular}

%% file: tables/rephrased.tex
\begin{tabular}{@{}lrrrr@{}}
\toprule
& \multicolumn{2}{c}{Reviews} & \multicolumn{2}{c}{GoEmotions} \\
\cmidrule(lr){2-3}\cmidrule(l){4-5}
Method & Calls & Miss & Calls & Miss \\
\midrule
SemViews & 0.53 & 0.0\% & 0.47 & 0.2\% \\
SemViews, no views & 0.73 & 0.7\% & 0.74 & 0.5\% \\
BARGAIN & 0.91 & 0.3\% & 0.97 & 0.5\% \\
LOTUS cascade & 0.57 & \textbf{43.3\%} & 0.82 & \textbf{26.9\%} \\
SemWeave (reimpl.) & 0.07 & \textbf{68.3\%} & 0.06 & \textbf{82.5\%} \\
Similarity cache (0.8) & 0.38 & \textbf{52.9\%} & 0.24 & \textbf{63.0\%} \\
Similarity cache (0.9) & 0.55 & \textbf{46.2\%} & 0.79 & \textbf{27.6\%} \\
\bottomrule
\end{tabular}

%% file: tables/e9_overhead.tex
\begin{tabular}{lrrrr}
\toprule
Dataset & Preds & No views & SemViews & Overhead \\
\midrule
PubMed & 8 & 22,818 & 22,823 & +0.0\% \\
GoEmotions & 8 & 34,055 & 32,895 & -3.4\% \\
DBpedia & 8 & 25,848 & 25,400 & -1.7\% \\
Reviews & 8 & 26,500 & 25,431 & -4.0\% \\
\bottomrule
\end{tabular}

%% file: tables/e7_robustness.tex
\begin{tabular}{lrrrr}
\toprule
& \multicolumn{2}{c}{Same-oracle views} & \multicolumn{2}{c}{Views from another oracle} \\
\cmidrule(lr){2-3}\cmidrule(lr){4-5}
Method & Calls & Met & Calls & Met \\
\midrule
SemViews & 32,909 & 99\% & 29,526 & 100\% \\
SemViews, no views & 36,026 & 99\% & 36,026 & 99\% \\
LOTUS cascade & 30,276 & 63\% & 30,276 & 63\% \\
Similarity cache (0.9) & 30,276 & 63\% & 0 & 0\% \\
\bottomrule
\end{tabular}

%% file: tables/ablation.tex
\begin{tabular}{@{}lrrrr@{}}
\toprule
Variant & PubMed & GoEmo. & DBpedia & Reviews \\
\midrule
first draft (no refit, CP) & 1.23\,{\scriptsize[1.18,1.27]} & 1.16\,{\scriptsize[1.14,1.19]} & 1.27\,{\scriptsize[1.19,1.33]} & 1.08\,{\scriptsize[1.04,1.12]} \\
no views & 1.04\,{\scriptsize[1.00,1.08]} & 1.32\,{\scriptsize[1.29,1.35]} & 1.04\,{\scriptsize[0.99,1.16]} & 1.03\,{\scriptsize[1.02,1.08]} \\
no views, no embeddings & 1.50\,{\scriptsize[1.48,1.55]} & 1.61\,{\scriptsize[1.57,1.66]} & 1.67\,{\scriptsize[1.57,1.76]} & 1.22\,{\scriptsize[1.19,1.27]} \\
per-signature strata & 1.61\,{\scriptsize[1.57,1.66]} & 1.42\,{\scriptsize[1.35,1.44]} & 1.64\,{\scriptsize[1.52,1.71]} & 1.23\,{\scriptsize[1.21,1.29]} \\
no row embeddings & 1.38\,{\scriptsize[1.35,1.40]} & 1.18\,{\scriptsize[1.15,1.22]} & 1.51\,{\scriptsize[1.43,1.57]} & 1.15\,{\scriptsize[1.10,1.19]} \\
no helper scores & 1.11\,{\scriptsize[1.06,1.14]} & 1.01\,{\scriptsize[0.99,1.04]} & 0.99\,{\scriptsize[0.90,1.04]} & 0.97\,{\scriptsize[0.93,1.00]} \\
all views as candidates & 1.00\,{\scriptsize[0.98,1.05]} & 0.99\,{\scriptsize[0.96,1.01]} & 1.01\,{\scriptsize[0.95,1.06]} & 0.98\,{\scriptsize[0.95,1.03]} \\
shaping sample 100 & 1.02\,{\scriptsize[0.98,1.06]} & 1.02\,{\scriptsize[0.99,1.04]} & 1.02\,{\scriptsize[0.94,1.11]} & 1.00\,{\scriptsize[0.96,1.06]} \\
shaping sample 400 & 1.01\,{\scriptsize[0.97,1.05]} & 1.02\,{\scriptsize[1.00,1.05]} & 0.99\,{\scriptsize[0.93,1.09]} & 0.99\,{\scriptsize[0.95,1.02]} \\
\bottomrule
\end{tabular}

%% file: sections/discussion.tex
\section{Discussion and Lessons}
\label{sec:discussion}

\paragraph{What limits the savings?}
Two factors bound the savings: the certification floor, and how well past answers predict new
ones. To measure the floor of Proposition~\ref{prop:floor}, we replay every workload with a \emph{perfect} score,
the oracle label itself, while certification samples as usual. \sys then still needs
\claimHeadroomPerfectShare{} of the calls of exact evaluation on average, against
\claimHeadroomOperatorShare{} with its learned score: \claimHeadroomPerfectRareShare{} on predicates with
at most 2\% selectivity and \claimHeadroomPerfectCommonShare{} on those above 10\%. About half of the
remaining cost is thus a property of our 5{,}000-tuple relations, not of the score. Because the
floor does not grow with the relation, its share shrinks on larger tables. To vary only the size,
we simulate larger relations by repeating each relation's tuples 4 and 20 times, which
preserves every label, score, and embedding and each predicate's selectivity.
The share of exact evaluation falls from \claimHeadroomOperatorShare{} to
\claimHeadroomScaleMidShare{} at 20{,}000 tuples and \claimHeadroomScaleLargeShare{} at 100{,}000
(\claimHeadroomScaleLargeShapedShare{} with a shaping sample of 0.8\% of the table), with every miss
rate below $\delta$, while the perfect-score cost drops to \claimHeadroomScaleLargePerfectShare{}. The
simulation cannot tell whether real large tables are easier or harder to score, but it shows where
the limit moves: on large relations the floor nearly vanishes and almost all of the remaining cost
is the score's uncertainty, which refinement reduces and related views can reduce further.

The second factor is the predictive power of past answers. Even with every other predicate's
labels known exactly, a logistic model separates the predicates of a pool with a median AUROC of
\claimEOneAurocPubmed{} on PubMed, \claimEOneAurocReviews{} on Reviews, \claimEOneAurocGoemotions{}
on GoEmotions, and \claimEOneAurocDbpedia{} on DBpedia. Real predicates share many positives and
negatives, but rarely so cleanly that a past answer settles a new one; this is precisely why
views must inform the order and not the answer.

\paragraph{When do views help, and how general is the gain?}
Views help substantially when the workload contains a related predicate and cost almost nothing
otherwise (Sections~\ref{sec:eval-cost} and~\ref{sec:eval-robust}), so how much a whole workload
gains is a property of the workload, not of the operator. Our pools are a mixed test.
GoEmotions, where views help most, is generated from a taxonomy whose group predicates (Ekman
groups, sentiments) are near-unions of other predicates in the pool, and it was one of the two
pilot datasets used to fix hyperparameters; its gain is likely an upper end. PubMed's
human-written external pool has few related predicates and gains little. Even a perfectly warm
catalog, holding the complete oracle labels of every other predicate before each query, saves only
\claimHeadroomWarmExtraMin{}--\claimHeadroomWarmExtraMax{} more calls than the workloads we replay,
most on GoEmotions: on these pools, there is little more for views to give. Analyst and agent
sessions that refine a question step by step may hold more related predicates than a random order
of a pool, but we have not measured such sessions. The gains of the certifier and of free row embeddings over per-query baselines, in contrast,
appear on all four datasets.

\paragraph{Limitations.}
Our main relations have 5{,}000 tuples; the tiling simulation and larger-scale
experiments (Appendix~\ref{app:scale}) show savings growing with the relation, but we have
not measured relations of millions of tuples. Each workload issues every predicate of a pool
in a random order, and two pools are generated from taxonomies; real analyst and agent sessions
and longer-lived catalogs remain to be measured. The two datasets used to fix
hyperparameters, GoEmotions and PubMed, are not fully held out (Appendix~\ref{app:repro}). Finally,
we study filters only.

%% file: sections/related.tex
\section{Related Work}
\label{sec:related}

\paragraph{Semantic operator systems.}
\lotus~\cite{patel2024semops,patel2025lotus} defines semantic operators with a reference-algorithm
semantics, which we adopt. Palimpzest~\cite{liu2024palimpzest} and Abacus~\cite{russo2026abacus}
choose models, prompts, and implementations under cost and quality constraints;
DocETL~\cite{shankar2024docetl} rewrites document pipelines with an agent; UQE~\cite{dai2024uqe}
combines sampling and LLM calls for unstructured queries; Nirvana~\cite{zhu2025nirvana},
CADENZA~\cite{ha2026cadenza}, and PLOP~\cite{mang2026plop} optimize the logical and physical plan
of one query; SEMA-SQL~\cite{lin2026semasql} generates and optimizes hybrid queries; and Cortex
AISQL~\cite{liskowski2026cortex} brings AI-aware optimization and cascades into a production SQL
engine. QUEST~\cite{sun2025quest} reduces the tokens of each extraction call. These systems optimize one
query at a time and do not use the oracle's labels for one predicate to answer another. \sys is a
plug-in for such systems: it changes how one filter is executed and keeps the labels for later
queries.

\paragraph{Per-query optimization with guarantees.}
The \lotus cascade~\cite{patel2025lotus} learns thresholds on a helper score for precision and
recall targets with normal-approximation bounds adapted from SUPG~\cite{kang2020supg};
BARGAIN~\cite{zeighami2025bargain} replaces them with betting-based
bounds~\cite{waudbysmith2020betting} and adaptive sampling, and streaming
cascades~\cite{liskowski2026streaming} refine thresholds as labels arrive.
ScaleDoc~\cite{zhang2025scaledoc}, CSV~\cite{hou2026csv}, Larch~\cite{zhao2026larch}, Task
Cascades~\cite{shankar2026taskcascades}, lightweight proxies~\cite{chung2026proxy}, and a unified
template for semantic filtering~\cite{kim2026fastfilter} cut the cost of one filter further, and
FDJ~\cite{zeighami2025fdj} prunes semantic joins with guarantees. Query-agnostic embedding indexes
such as TASTI~\cite{kang2022tasti} are reused across queries, as are thresholds and proxies for the
\emph{same} predicate in \lotus's \texttt{CascadeOptimizer}~\cite{lotus2026software} and in offline
proxies~\cite{chung2026proxy}. Stretto~\cite{sanmartino2026stretto} and
ConRAD~\cite{horchidan2026conrad} allocate error budgets across the operators of one plan. All of
these pay for every label of the current predicate. \sys is complementary: it shrinks the set on
which a per-query method has to work.

\paragraph{Reuse across predicates and queries.}
SemWeave~\cite{mahmood2026semweave} is closest to our goal. It intercepts semantic filters in
Palimpzest and \lotus, asks an LLM for the entailment relation between a new filter and similar
past ones, and derives tuples' labels from stored results when the relation determines them,
treating a prior result as a materialized semantic view. It checks relations for logical
consistency but offers no guarantee on the result. EVA~\cite{xu2022eva} reuses the outputs of the
same expensive UDF across exploratory queries whose predicates overlap symbolically. Both trust the
inferred relation; we show that judged relations often fail on the oracle's labels and that
statistical containment degrades along chains, and \sys certifies every reuse instead. Exact-match
caches~\cite{dorbani2025flockmtl,kumarasinghe2026ipdb} and memoization
layers~\cite{liskowski2026compositional} hit only on an identical prompt and tuple; Sema~\cite{qi2026sema}
deduces relational constraints from a single predicate. Similarity caches~\cite{gptcache,
regmi2024gptsemcache,gill2025meancache} reuse answers to nearby prompts without guarantees, and
work on their eviction policies~\cite{liu2025semcacheserving} is relevant to view retention.
vCache~\cite{schroeder2025vcache} bounds the per-response error of such a cache, not the precision
and recall of a result set. Semantic caching for OLAP~\cite{bindschaedler2026olapcache} reuses
results through exact intent matches and logical derivations over structured queries.
SemBaker~\cite{dong2026sembaker} compiles semantic operators into generated code that replaces the
oracle, a different trade-off from reusing oracle labels.

\paragraph{Views, caching, and multi-query optimization.}
Answering queries using views~\cite{halevy2001views}, view matching in
optimizers~\cite{goldstein2001matviews}, view selection~\cite{harinarayan1996cubes}, semantic data
caching~\cite{dar1996semantic}, and multi-query optimization~\cite{sellis1988mqo} all rely on
logical containment between predicates over structured data. Semantic predicates are related only
statistically, which is why \sys uses views to plan but not to decide. Probabilistic
predicates~\cite{lu2018probabilistic} train cheap filters ahead of time for one predicate, and
semantic integrity constraints~\cite{lee2025sic} declare and check constraints on semantic
operator outputs; logical relations between predicates could be stated as such constraints but
would still need certification.

\paragraph{Sampling, estimation, and statistical inference.}
Stratified sampling for approximate query
processing~\cite{chaudhuri2007stratified} chooses strata from the workload to
answer aggregates; ABae~\cite{kang2021abae} stratifies by a proxy score for aggregates with
expensive predicates; OLLA~\cite{hui2026olla} and query-centric AI workflows~\cite{wang2026querycentric}
aggregate over LLM outputs online with confidence intervals. Prediction-powered
inference~\cite{angelopoulos2023ppi}, its stratified variant~\cite{fisch2024stratppi}, and active
statistical inference~\cite{zrnic2024active} combine model predictions with labeled samples for
valid estimates. Semantic histograms~\cite{urban2026semhist} and the SemCEB
benchmark~\cite{zimmerer2026semceb} estimate the selectivity of semantic predicates. These methods
estimate aggregates or selectivities; \sys certifies the precision and recall of a set-valued
answer.
Benchmarks of semantic queries such as SemBench~\cite{lao2025sembench} could supply further
workloads.

%% file: sections/conclusion.tex
\section{Conclusion}
\label{sec:conclusion}

\sys shows that a semantic filter can be both certified and cheap. By ordering tuples with side
information that costs no oracle calls, refining the order with labels it must pay for anyway, and
certifying only with fresh samples, it saves several times as many calls as the strongest prior
per-query method with the same guarantee. The same separation makes the oracle's labels safely
reusable across queries: views help in proportion to how related a workload's predicates are, most
when a question returns in other words, whereas uncertified reuse by similarity or entailment
misses its targets on a large share of queries. More broadly, letting side information shape the
plan while only fresh samples enter a bound lets a semantic engine exploit any knowledge it has
accumulated without risking its guarantees. Extending certified reuse to semantic joins and
ranking, and measuring how often real sessions revisit their questions, remain open.

%% file: appendix.tex
\section{Proofs}
\label{app:proofs}

\subsection{Proof of Proposition~\ref{lem:chain}}
\label{app:chain}
Let $P,Q,\mathcal{R}$ be the sets of tuples on which the oracle labels $p$, $q$, $r$ positive, with
$|P|>0$, $|Q|>0$. Assume $|P\setminus Q|\le\eps_1|P|$ and $|Q\setminus \mathcal{R}|\le\eps_2|Q|$. Then
\[
|P\setminus \mathcal{R}| \;=\; |(P\setminus Q)\setminus \mathcal{R}| + |(P\cap Q)\setminus \mathcal{R}|
\;\le\; |P\setminus Q| + |Q\setminus \mathcal{R}|
\;\le\; \eps_1|P| + \eps_2|Q|,
\]
and dividing by $|P|$ gives $\Pr[r{=}0\mid p{=}1]\le \eps_1+\eps_2\,\sel(q)/\sel(p)$.

\emph{Tightness.} Fix $|P|$, $|Q|$, $\eps_1$, $\eps_2$ with
$\eps_2|Q|\le(1-\eps_1)|P|$ (the condition of the proposition; $|Q|\ge(1-\eps_1)|P|$ then holds because $Q\supseteq P\cap Q$). Choose $P\setminus Q$ of size
$\eps_1|P|$, choose $Q\supseteq P\cap Q$ with $|P\cap Q|=(1-\eps_1)|P|$, and let $\mathcal{R}$ exclude
exactly $P\setminus Q$ and $\eps_2|Q|$ tuples of $P\cap Q$ (possible by the first condition),
and include every other tuple. Then $p\contained{\eps_1}q$ and $q\contained{\eps_2}r$ hold with
equality, and $|P\setminus \mathcal{R}|=\eps_1|P|+\eps_2|Q|$, so both inequalities above are equalities.

\subsection{Proof of Theorem~\ref{thm:valid}}
\label{app:valid}
\paragraph{Setting.}
Fix the relation, the oracle labels, the catalog, the helper scores, and the embeddings. Let
$\mathcal{F}$ denote everything \textsc{CertifiedFilter} observes before its certification
draws: the catalog, the shaping sample, the refinement tuples and their labels, and the
fitted score $\hat\pi$. Refinement is part of $\mathcal{F}$: it labels tuples, refits, and re-plans,
but all of it happens before the regions are fixed and the certification draws are made.
Conditional on $\mathcal{F}$, the regions $X$, $A$, $O$ and the sample sizes $n_X$, $n_A$ are
fixed, and so is the split of $\delta$ ($\alpha_X+\alpha_A=\delta$, or $\delta$ when only one
region is non-empty). The certification draws are made after $\mathcal{F}$ with fresh randomness:
$n_X$ tuples uniformly without replacement from $X$ and $n_A$ from $A$. Conditional on
$\mathcal{F}$, the number of positives among the draws from $X$ is hypergeometric with parameters
$(|X|, c_X, n_X)$, where $c_X$ is the (fixed) number of oracle positives in $X$; likewise for $A$.

\paragraph{Coverage event.}
Let $u_X=\max\{c: \Pr_{\mathrm{HG}(|X|,c,n_X)}[x\le x_X]>\alpha_X\}$ and
$\ell_A=\min\{c: \Pr_{\mathrm{HG}(|A|,c,n_A)}[x\ge x_A]>\alpha_A\}$ be the one-sided exact
hypergeometric bounds computed from the numbers of positives $x_X$ and $x_A$ among the draws. By the
standard test-inversion argument (the tail probabilities are monotone in $c$), conditional on
$\mathcal{F}$, $\Pr[c_X>u_X]\le\alpha_X$ and $\Pr[c_A<\ell_A]\le\alpha_A$. Let $E=\{c_X\le u_X\}\cap\{c_A\ge\ell_A\}$.
A union bound gives $\Pr[E^c\mid\mathcal{F}]\le\delta$, hence $\Pr[E^c]\le\delta$ unconditionally.
No other quantity computed later enters a bound.

\paragraph{Deterministic step.}
Gap closing only adds oracle labels of tuples in $X$ or $A$; it draws nothing and recomputes no
bound. Consider any state the algorithm can stop in. Let $P$ be the number of known positives
outside $A$ (in $K$, which includes the shaping and refinement labels, in $O$, and those found
in $X$), $d_X$ and $\nu_X$
the known positives and the still-unknown tuples in $X$, $d_A$ and $e_A$ the known positives and
negatives in $A$. The returned set is all known positives plus $A$ without its known negatives.
Its true positives are $P + c_A$ (no positive of $A$ is dropped), its size is $P+|A|-e_A$, and the
positives it misses are the unknown positives of $X$, numbering $c_X-d_X\le\nu_X$. Therefore
\[
\mathrm{rec}=\frac{P+c_A}{P+c_A+(c_X-d_X)},\qquad \mathrm{prec}=\frac{P+c_A}{P+|A|-e_A}.
\]
Both are non-decreasing in $c_A$, and the recall is non-increasing in $c_X$. On $E$,
$c_A\ge\max(\ell_A,d_A)$ and $c_X-d_X\le\min(\max(u_X-d_X,0),\nu_X)$; substituting yields exactly the
two quantities that the stopping rule requires to reach $\gamma_R$ and $\gamma_P$. So on $E$, every
state in which the algorithm stops meets both targets, and the algorithm always stops (labeling
$X$ and $A$ completely sets $\nu_X=0$ and drops every negative of $A$, so both quantities equal $1$). Hence $\Pr[\text{targets met}]\ge\Pr[E]\ge1-\delta$.

\paragraph{What the proof does not use.}
It uses no property of the views, the helper, the embeddings, or the fitted score; they only
determine $X$, $A$, $O$, $n_X$, $n_A$ through $\mathcal{F}$. The shaping labels are known exactly and
contribute only to $P$, and so do the refinement labels. Sampling without replacement makes the
draws exactly hypergeometric for the finite region, so no approximation is involved; with
replacement and Clopper--Pearson bounds, as in the ablation ``first draft'', the same argument
holds with binomial draws. The guarantee is per
query: for each query in a workload, conditional on everything that happened before it, the
probability of missing a target is at most $\delta$. The miss indicators of a workload therefore
form a sequence whose conditional means are at most $\delta$, so by the Azuma--Hoeffding
inequality the miss rate of a long workload exceeds $\delta$ by more than $O(1/\sqrt{n})$ only with
small probability; the validity experiment of Section~\ref{sec:eval-validity} measures this rate.

\paragraph{General strata.}
The same argument extends to $m$ predefined strata decided individually as accept, reject, or
oracle, with simultaneous bounds from a union bound over $2m$ one-sided bounds (and over a
predefined schedule of $L$ looks, at level $\delta w_i/(2L)$ per bound with $\sum_i w_i\le1$):
the recall and precision ratios are monotone in every stratum's positive count, so plugging in
the bounds is valid for every assignment simultaneously. The ablation ``per-signature strata'' in
Section~\ref{sec:eval-sensitivity} is this variant.

\section{Larger-Scale Experiments}
\label{app:scale}
Our main relations have 5{,}000 tuples, so certification samples are a large part of each query
(Section~\ref{sec:discussion}). We chose that size to limit the cost and carbon footprint of the
experiments: measuring validity requires the oracle's label for every tuple and predicate, so the
labeling cost grows with the number of tuples times the number of predicates. To verify that our findings hold on larger
relations without labeling every predicate, we built two real relations of
\claimEOneThreeRows{} tuples from the sources of the main tables, disjoint from them: DBpedia
articles from the training split (the 5{,}000-tuple table comes from the test split), and Amazon
Electronics reviews from the start of the same source file, without the reviews of the
5{,}000-tuple table. For each relation we chose \claimEOneThreeLlmDbpediaPreds{} predicates of its
pool at random with a fixed seed, had the oracle label every tuple (\claimEOneThreeLlmLabels{}
labels in total; selectivity \claimEOneThreeLlmDbpediaSelMin{}--\claimEOneThreeLlmDbpediaSelMax{} on
DBpedia and \claimEOneThreeLlmReviewsSelMin{}--\claimEOneThreeLlmReviewsSelMax{} on Reviews), and
replayed \claimEOneThreeLlmDbpediaOrders{} random orders of each pool. On DBpedia we also ran all
\claimEOneThreeHumanPreds{} predicates with the human ontology labels as the oracle
(\claimEOneThreeHumanOrders{} orders). That oracle is exact and free but is not the LLM: on the
5{,}000-tuple table the two agree on \claimEOneThreeHumanLlmAgreeMedian{} of tuples (median over
predicates) with a median F1 of only \claimEOneThreeHumanLlmFOneMedian{}, so we use it only as a
second, larger workload. The first 5{,}000 and 20{,}000 tuples of each relation, with the same
labels, serve as the smaller tables, so within each block of Table~\ref{tab:e13} only the table
size changes. The operator is the paper's with two changes: the helper is off, because scoring
every tuple--predicate pair with it locally would take more than a day (without it, \sys needs no
more calls on the 5{,}000-tuple DBpedia and Reviews tables; Table~\ref{tab:ablation}), and above
5{,}000 tuples the certification sample may be as large as in the tiling simulation. Nothing was
tuned on these data.

\begin{table}[t]
\caption{\sys on relations of up to 100{,}000 tuples: oracle calls as a share of exact evaluation
(pooled over queries and orders), for all predicates and for those with selectivity at most 2\%,
and the share of queries meeting both targets ($p=r=0.9$, $\delta=0.05$). Smaller tables are
prefixes of the same relation. Helper off.}
\label{tab:e13}
\centering
\begin{adjustbox}{max width=\columnwidth}\input{tables/e13_scale}\end{adjustbox}
\end{table}

\paragraph{Results.}
\sys stays valid at every size: the share of queries that miss a target never exceeds
\claimEOneThreeLlmDbpediaMaxMiss{} on DBpedia and \claimEOneThreeLlmReviewsMaxMiss{} on Reviews with
the LLM oracle, and \claimEOneThreeHumanMaxMiss{} with human labels, against $\delta=5\%$. Its
relative cost falls as the table grows: from \claimEOneThreeLlmDbpediaShareSmall{} of exact
evaluation at 5{,}000 tuples to \claimEOneThreeLlmDbpediaShareLarge{} at 100{,}000 on DBpedia, from
\claimEOneThreeLlmReviewsShareSmall{} to \claimEOneThreeLlmReviewsShareLarge{} on Reviews, and from
\claimEOneThreeHumanShareSmall{} to \claimEOneThreeHumanShareLarge{} with human labels. Rare
predicates, whose cost on small tables is mostly the certification floor, gain most: on Reviews
their share drops from \claimEOneThreeLlmReviewsRareShareSmall{} to
\claimEOneThreeLlmReviewsRareShareLarge{}. The real tables thus confirm the trend of the tiling
simulation, which predicted \claimHeadroomScaleLargeShare{} at 100{,}000 tuples on average over the
four datasets; DBpedia falls below that and Reviews stays above it. Views add nothing here:
predicates drawn at random from a pool of 40 are rarely related, as in the unrelated control of
Section~\ref{sec:eval-robust}, so the gain at scale comes from the certifier and the learned score.

\section{Reproducibility}
\label{app:repro}
All results replay recorded oracle labels; the scripts, configurations, and recorded labels needed to
regenerate every figure and table are at \url{https://github.com/semviews/lotus-semviews}. Hyperparameters (shaping-sample size 200, regularization, number of candidate views, region
grid) were fixed on synthetic data and on partial label records of GoEmotions and PubMed before the full
runs, and frozen in the experiment configuration. The refinement rounds, their batch, the
planning margin, the $\delta$ split, and sampling without replacement were then chosen in three
tuning rounds on GoEmotions and PubMed only, each round on fresh workload orders (orders that the
reported runs on these two datasets also include), with a check on GoEmotions and PubMed tiled 4 and 20 times
that led to a batch proportional to the relation. The two tuning datasets are therefore not fully
held out; DBpedia and Reviews are. The ablation reports the first-draft operator, without
refinement and with Clopper--Pearson bounds, for comparison.

\section{AI-use disclosure}
Generative AI tools were used to write parts of the code and text under the authors'
direction; the details required by the ACM policy are listed in the repository's
\texttt{docs/AI\_USE.md}.

%% file: tables/e13_scale.tex
\begin{tabular}{@{}lrrrrr@{}}
\toprule
& & \multicolumn{2}{c}{Share of exact} & Rare & Both \\
\cmidrule(lr){3-4}
Oracle & Tuples & \sys & no views & ($\le$2\%) & met \\
\midrule
DBpedia, human labels & 5,000 & 52\% & 58\% & 64\% & 100.0\% \\
(40 predicates) & 20,000 & 31\% & 35\% & 40\% & 100.0\% \\
 & 100,000 & 20\% & 21\% & 23\% & 99.4\% \\
\midrule
DBpedia, LLM oracle & 5,000 & 60\% & 60\% & 67\% & 100.0\% \\
(8 predicates) & 20,000 & 45\% & 46\% & 68\% & 100.0\% \\
 & 100,000 & 34\% & 34\% & 34\% & 99.4\% \\
\midrule
Reviews, LLM oracle & 5,000 & 77\% & 77\% & 95\% & 99.4\% \\
(8 predicates) & 20,000 & 64\% & 63\% & 76\% & 100.0\% \\
 & 100,000 & 55\% & 54\% & 66\% & 100.0\% \\
\bottomrule
\end{tabular}